\documentclass[a4paper, 11pt]{article}

\usepackage{bbold}
\DeclareSymbolFont{bbgreek}{U}{bbold}{m}{n}
\DeclareMathSymbol{\bbSigma}{\mathbb}{bbgreek}{'6}

\usepackage{amsmath, amssymb, amsthm}
\usepackage{graphicx}
\usepackage{color}
\usepackage{multirow}
\usepackage{ltxtable}
\usepackage[dvips]{hyperref}
\usepackage[hang]{subfigure}

\newtheorem{theorem}{Theorem}
\newtheorem{lemma}{Lemma}
\newtheorem{definition}{Definition}
\newtheorem{example}{Example}
\newtheorem{property}{Property}

\newcommand\bbR{\mathbb{R}}
\newcommand\bbN{\mathbb{N}}
\newcommand\bxi{\boldsymbol{\xi}}
\newcommand\bx{\boldsymbol{x}}

\newcommand\bu{\boldsymbol{u}}
\newcommand\bv{\boldsymbol{v}}

\newcommand\bn{\boldsymbol{n}}
\newcommand\br{\boldsymbol{r}}
\newcommand\bz{\boldsymbol{z}}
\newcommand\bA{{\bf{A}}}
\newcommand\bB{{\bf{B}}}
\newcommand\bF{\boldsymbol{F}}

\newcommand\bPhi{\boldsymbol{\Phi}}
\newcommand\dd{\,\mathrm{d}}
\newcommand\He{\mathit{He}}

\newcommand\bw{\boldsymbol{w}}
\newcommand\RM{{\cal R}_{M,D}}
\newcommand\tal{\tilde{\alpha}}
\newcommand\cS{{\cal{S}}}
\newcommand\rC[2]{{\rm{C}}_{{#1},{#2}}}
\newcommand\rmC{{\rm{C}}}

\newcommand\bAM{{{\bf{A}}_M}}
\newcommand\bhAM{{\hat{\bf{A}}_M}}
\newcommand\btAM{{\tilde{\bf{A}}_M}}
\newcommand\pd[2]{\dfrac{\partial {#1}}{\partial {#2}}}

\newcommand\mN{{\mathcal N}_{D}}

\numberwithin{equation}{section}

\newcommand\comment[1]{}

\graphicspath{{images/}}

{\theoremstyle{remark} \newtheorem{remark}{Remark}}

\title{Globally Hyperbolic Regularization of Grad's Moment System}

\author{Zhenning Cai\thanks{School of Mathematical Sciences, Peking
    University, Beijing, China, email: {\tt caizn@pku.edu.cn}.},~~
  Yuwei Fan\thanks{School of Mathematical Sciences, Peking University,
    Beijing, China, email: {\tt ywfan@pku.edu.cn}.},~~ Ruo
  Li\thanks{CAPT, LMAM \& School of Mathematical Sciences, Peking
    University, Beijing, China, email: {\tt rli@math.pku.edu.cn}.}}

\begin{document}
\maketitle
\begin{abstract}
  In this paper, we propose a globally hyperbolic regularization to
  the general Grad's moment system in multi-dimensional spaces.
  Systems with moments up to an arbitrary order are studied. The
  characteristic speeds of the regularized moment system can be
  analytically given and only depend on the macroscopic velocity and
  the temperature. The structure of the eigenvalues and eigenvectors
  of the coefficient matrix is fully clarified. The regularization
  together with the properties of the resulting moment systems is
  consistent with the simple one-dimensional case discussed in
  \cite{Fan}. Besides, all characteristic waves are proven to be
  genuinely nonlinear or linearly degenerate, and the studies on the
  properties of rarefaction waves, contact discontinuities and shock
  waves are included.

\vspace*{4mm}
\end{abstract}

\section{Introduction}
The kinetic gas theory, which is based on the Boltzmann equation, is
one of the fundamental tools in modelling non-equilibrium
processes. Nevertheless, in most cases, a direct numerical
discretization of the Boltzmann equation leads to unacceptable
computational costs. In 1940s, Grad \cite{Grad} proposed the moment
approximation of the distribution function, trying to establish a
series of intermediate models between the fluid dynamics and the
kinetic theory. However, due to a number of defects in Grad's
13-moment equations, such as the appearance of unphysical subshocks,
nonexistence of an entropy function, and lack of global hyperbolicity,
not much attention is paid to the moment method in the last century.

In the recent twenty years, as the investigation into the moment
method becomes deeper, various ``regularizations'' are proposed to
challenge the traditional accusations on the moment method. A list of
relevant publications can be found in the references of
\cite{TorrilhonEditorial}. Recently, we are interested in the large
moment system together with its numerical methods \cite{NRxx,
  NRxx_new, Cai, Li}, and it is found that the lack of the
well-posedness due to the loss of global hyperbolicity is a major
obstacle in our simulations, especially for large Mach number gas
flows \cite{NRxx_new}. Torrilhon \cite{Torrilhon2010} provided a
13-moment hyperbolic moment system based on multi-variate Pearson-IV
distributions, but it seems unlikely to extend the same technique to
systems with large number of moments. As discussed in
\cite{TorrilhonEditorial}, Levermore \cite{Levermore} gave a partial
answer to the question of hyperbolicity of large moment system based
on a maximal entropy distribution function. However, the analytical
forms of Levermore's equations cannot be obtained once the number of
moments is greater than $10$. While exploring the method ensuring the
hyperbolicity of the moment system, we discovered \cite{Fan} that the
structure of the characteristic polynomial of Grad's moment equations
with one-dimensional microscopic velocity is rather simple; thus a
globally hyperbolic regularization can be achieved by simply adding
two terms to the equation of the highest order moment.

In this paper, the results in \cite{Fan} are extended to the
multi-dimensional space. For multi-dimensional moment systems, the
regularization method is consistent with the one-dimensional case.
Due to the complexity of the moment systems, this paper is mainly
devoted to a rigorous proof of the hyperbolicity of regularized moment
system for any space dimensions and an arbitrary order of moments. The
result is obtained by firstly restricting the spatial variable in
the one-dimensional space, and then it is generalized to the
multi-dimensional space using the rotation invariance of the
regularized system. For the case of one-dimensional spatial variable,
the structure of the coefficient matrix is similar as the Hessenberg
matrix, which enables us to calculate the eigenvectors for a given
eigenvalue. Then, the hyperbolicity of the moment system follows by
counting the number of linearly independent eigenvectors. At the same
time, the expressions of all characteristic speeds are obtained, each
of which is a sum of the macroscopic velocity and the square root of
the temperature scaled by a zero of the Hermite polynomial. Besides,
we prove that each characteristic filed of the hyperbolic moment
systems is either genuinely nonlinear or linearly degenerate, and some
properties of the rarefaction waves, the contact discontinuities and
the shock waves are investigated.

The rest of this paper is arranged as follows: in Section
\ref{sec:NRxx_method}, a brief review on the moment methods of
Boltzmann equation and the results in \cite{Fan} are presented.
Section \ref{sec:hyperbolicms} gives the globally hyperbolic
regularization for moment system with one-dimensional spatial variable
and multi-dimensional microscopic velocities. And in Section
\ref{sec:hyperbolicitymd}, the result for full multi-dimensional
moment system is proved. The study on the characteristic waves is
carried out in Section \ref{sec:rp}.


\section{Preliminaries} \label{sec:NRxx_method}

In this section, a concise introduction of the Boltzmann equation is
presented. And then some results of the work on the moment method in
\cite{NRxx_new, Fan} are briefly reviewed.

\subsection{Moment methods for Boltzmann equation}
\label{sec:boltzmann_equations}
Let the motion of particles be depicted by the distribution function
$f(t, \bx, \bxi)$ governed by the Boltzmann transport equation
\begin{equation} \label{eq:Boltzmann}
\pd{f}{t} +
  \sum_{j=1}^D\xi_j \pd{f}{x_j} = Q(f,f),
\qquad t \in \bbR^+, \quad \bx, \bxi \in \bbR^D,
\end{equation}
where $t$ denotes the time, $\bx = (x_1, \cdots, x_D)$ and $\bxi =
(\xi_1, \cdots, \xi_D)$ stand for the spatial coordinates and the
microscopic velocity, respectively. The right hand side $Q(f,f)$ is
the collision term describing the interaction between particles. In
this paper, we are focusing on the transportation part, thus the
collisionless Boltzmann equation with vanished $Q(f,f)$ is considered.

The moment method proposed by Grad \cite{Grad} approximates the
distribution function by a finite set of moments. To achieve this, we
expand $f$ into the Hermite series as in \cite{NRxx}:
\begin{equation} \label{eq:expansion}
f(t,\bx,\bxi) = \sum_{\alpha \in \bbN^D}f_{\alpha}(t,\bx)
  \mathcal{H}_{\theta(t,\bx),\alpha}
  \left( \frac{\bxi - \bu(t,\bx)}{\sqrt{\theta(t,\bx)}} \right),
\end{equation}
where $\alpha = (\alpha_1, \cdots, \alpha_D)$ is a $D$-dimensional
multi-index, and the basis functions are defined as
\begin{equation} \label{eq:cal_H}
\mathcal{H}_{\theta,\alpha}(\bz) =
  \prod_{d=1}^D \frac{1}{\sqrt{2\pi}}\theta^{-\frac{\alpha_d+1}{2}}
  \He_{\alpha_d}(z_d) \exp \left( -\frac{z_d^2}{2} \right),
\quad \bz = (z_1, \cdots, z_D) \in \bbR^D,
\end{equation}
where $\He_k$ is the $k$-th degree Hermite polynomial:
\begin{equation} \label{eq:He}
\He_k(x) = (-1)^k \exp \left( \frac{x^2}{2} \right)
  \frac{\dd^k}{\dd x^k} \exp \left( -\frac{x^2}{2} \right),
  \quad k \in\bbN.
\end{equation}
In \eqref{eq:expansion}, $\bu(t, \bx) = (u_1(t, \bx), \cdots, u_D(t,
\bx))$ and $\theta(t, \bx)$ denote the macroscopic velocity and
temperature, respectively, and they are related to $f$ by
\begin{equation}\label{eq:def_rhoutheta}
\begin{aligned}
\rho(t,\bx) &= \int_{\bbR^D} f(t, \bx, \bxi) \dd \bxi,\\
\rho(t,\bx) \bu(t,\bx) &= \int_{\bbR^D}\bxi f(t,\bx,\bxi)\dd\bxi,\\
\rho(t,\bx) |\bu(t,\bx)|^2 + D \rho(t,\bx)\theta(t,\bx)
  &= \int_{\bbR^D} |\bxi|^2 f(t,\bx,\bxi)\dd\bxi,
\end{aligned}
\end{equation}
where $\rho$ stands for the density of the gas. The following
relations can be deduced from the orthogonality of Hermite
polynomials:
\begin{equation}
f_0 = \rho, \quad f_{e_j} = 0, \quad \sum_{d=1}^D f_{2e_d} = 0,
  \qquad j = 1, \cdots, D,
\end{equation}
where $e_j$ is the $D$-dimensional multi-index with its $j$-th
component to be the only nonzero one and equals to $1$.

The moment system has been deduced in \cite{NRxx_new}, and here we
directly present the result therein:
\begin{equation} \label{eq:momentseqs1}
\begin{split}
& \left( \frac{\partial f_{\alpha}}{\partial t} +
  \sum_{d=1}^D \frac{\partial u_d}{\partial t} f_{\alpha-e_{d}}
  + \frac{1}{2} \frac{\partial \theta}{\partial t}
    \sum_{d =1}^D f_{\alpha-2e_d}
\right) \\
& \qquad + \sum_{j=1}^D \left(
  \theta \frac{\partial f_{\alpha-e_j}}{\partial x_j} +
  u_j \frac{\partial f_{\alpha}}{\partial x_j} +
  (\alpha_j+1)
    \frac{\partial f_{\alpha+e_j}}{\partial x_j}
\right) \\
& \qquad {} + \sum_{j=1}^D \sum_{d=1}^D
\frac{\partial u_d}{\partial x_{j}} \left(
  \theta f_{\alpha-e_d-e_{j}} + u_{j} f_{\alpha-e_d}
  + (\alpha_{j}+1) f_{\alpha-e_d+e_{j}}
\right) \\
& \qquad {} + \frac{1}{2} \sum_{j=1}^D \sum_{d=1}^D
  \frac{\partial \theta}{\partial x_{j}}
  \left( \theta f_{\alpha-2e_d-e_{j}} +
  u_{j} f_{\alpha-2e_d} +
  (\alpha_{j}+1) f_{\alpha-2e_d+e_{j}}
\right) = 0, \quad \forall \alpha \in \bbN^D.
\end{split}
\end{equation}
In this equation, $f_{\beta}$ is taken as zero if any components of
$\beta$ is negative. Some special choices of $\alpha$ lead to the
classic hydrodynamic equations:
\begin{subequations}\label{eq:conservation_laws}
\begin{align}
&\pd{\rho}{t} + \sum_{j=1}^D\left(u_j\pd{\rho}{x_j} + 
  \rho \pd{u_j}{x_j}\right) = 0, \\
\label{eq:u}
& \rho \pd{u_i}{t} + \sum_{j=1}^D\left(\rho u_j \pd{u_i}{x_j} +
  \pd{p_{e_i+e_j}}{x_j}\right) = 0, \quad i=1,\cdots,D,\\
\label{eq:theta}
&\frac{D}{2}\rho \pd{\theta}{t}
    + \sum_{j=1}^D\left( \frac{D}{2} \rho u_j \pd{\theta}{x_j}
  + \pd{q_j}{x_j}\right) 
  + \sum_{i=1}^D\sum_{j=1}^Dp_{e_i+e_j} \pd{u_i}{x_j} = 0,
\end{align}
\end{subequations}
where $p_{e_i+e_j}$ is the pressure tensor\footnote{In some
  literatures, the pressure tensor is denoted as $p_{ij}$, $i,j =
  1,\cdots, D$. Here the special subscript $e_i + e_j$ is used to
  match the form of general moments $f_{e_i+e_j}$ for convenience in
  later use.} and $q_j$ is the heat flux. They are defined as
\begin{subequations}
\begin{align}
\label{eq:basic_para_pij}
p_{e_i+e_j} & = \int_{\bbR^D} (\xi_i - u_i) (\xi_j - u_j) f\dd\bxi
  = \delta_{ij} \rho \theta + (1 + \delta_{ij}) f_{e_i + e_j}, \\
q_j & = \frac{1}{2} \int_{\bbR} |\bxi - \bu|^2 (\xi_j-u_j) f \dd \bxi
  = 2 f_{3e_j} + \sum_{d=1}^D f_{e_j + 2e_d},
\end{align}
\end{subequations}
where $\delta$ is Kronecker's delta symbol.  We refer the readers to
\cite{NRxx_new} for the detailed derivation of
\eqref{eq:conservation_laws}.

Since \eqref{eq:momentseqs1} forms an infinite set of moment equations
which are not suitable for practical use, the moment closure is in
need. The simplest way is to select an integer $M \geqslant 3$ and
force $f_{\alpha} = 0$ if $|\alpha| > M$, and the result is the
Grad-type system with $\binom{M+D}{D}$ moments.

\subsection{Regularization with 1D velocity space}
It is well known that the lack of global hyperbolicity is one of the
major defects of Grad's moment equations. For the thirteen moment
case, the hyperbolicity region has been analytically obtained in
\cite{Muller}. The construction of globally hyperbolic moment systems
is very meaningful to the robustness of fluid simulation using moment
approximation. In this direction, a general method by Levermore in
\cite{Levermore} on the construction of symmetric hyperbolic moment
systems is proposed. Later, Torillhon \cite{Torrilhon2010} raises a
clever idea to enlarge the hyperbolicity region of the 13-moment
system by using Pearson-IV-distributions. In \cite{Fan}, we have
studied the general 1D moment systems and found a way to make globally
hyperbolic regularization based on the characteristic speed
correction. Here we are going to give a brief review on the results
therein.

When $D = 1$, the multi-index $\alpha$ becomes a natural number.
Substituting \eqref{eq:u} and \eqref{eq:theta} into
\eqref{eq:momentseqs1}, we can eliminate the time derivative of the
velocity and temperature. Thus the $M$-th order Grad's moment system
can be written in the form of a quasi-linear system
\begin{equation} \label{eq:quasilinear_1D}
\frac{\partial \bw}{\partial t} +
  {\bf A}(\bw) \frac{\partial \bw}{\partial x} = 0,
\end{equation}
where ${\bf A}$ is a matrix dependent on $\bw$, and
\begin{equation}
\bw = (\rho, u, \theta, f_3, \cdots, f_M).
\end{equation}
In \cite{Fan}, we have obtained the following results:
\begin{enumerate}
\item The characteristic polynomial of ${\bf A}(\bw)$ is
\begin{equation}
\begin{array}{rcl}
|\lambda {\bf I} - {\bf A}| &=& 
  \theta^{\frac{M+1}{2}} \He_{M+1} \left(
    \dfrac{\lambda - u}{\sqrt{\theta}}
  \right) \\
&-& \dfrac{(M+1)!}{2\rho} \left[
    \left( (\lambda - u)^2 - \theta \right) f_{M-1}
    + 2(\lambda - u) f_M
  \right].
\end{array}
\end{equation}
\item By adding the regularization term based on characteristic speed correction
\begin{equation}\label{eqRM1D}
\mathcal{R}_M = \frac{M+1}{2} \left(
  2 f_M \pd{u}{x} + f_{M-1} \pd{\theta}{x}\right)
\end{equation}
to the right hand side of the last equation of
\eqref{eq:quasilinear_1D}, the system is turned to be globally
hyperbolic and the eigenvalues of the regularized moment system are
\begin{equation}
u + \rC{j}{M+1} \sqrt{\theta}, \quad j = 1,\cdots,M+1,
\end{equation}
where $\rC{j}{k}$ is the $j$-th root of the Hermite polynomial
$\He_{k}(x)$, noticing that $\He_k(x), k\in\bbN$ has $k$ different
zeros, which read $\rC{1}{k}, \cdots, \rC{k}{k}$ and satisfy
$\rC{1}{k}< \cdots < \rC{k}{k}$.
\end{enumerate}
The second result gives a practical implementation of a globally
hyperbolic regularization.

\subsection{Reformulation of the moment system}\label{sec:momentsystem}
In order to facilitate the studying of the moment system when $D\ge2$,
we rewrite \eqref{eq:momentseqs1} in another form. Let $p =
\rho\theta$, then $p = \frac{1}{D}\sum_{d=1}^Dp_{2e_d}$, and we
have
\begin{equation}\label{eq:pdtheta}
\frac{\partial\theta}{\partial x_j} =
-\frac{\theta}{\rho}\pd{\rho}{x_j}
+\frac{1}{D\rho}\sum_{d=1}^D\pd{p_{2e_d}}{x_j},\quad j=1,\cdots,D.
\end{equation}
By substituting \eqref{eq:conservation_laws} and \eqref{eq:pdtheta}
into \eqref{eq:momentseqs1}, the following equation is obtained with
some simplification:
\begin{equation}\label{eq:momentseqs2}
\begin{split}
&\pd{f_{\alpha}}{t}+\sum_{j=1}^D\left(
    \theta\pd{f_{\alpha-e_j}}{x_j}+u_j\pd{f_{\alpha}}{x_j}
    +(\alpha_j+1)\pd{f_{\alpha+e_j}}{x_j}\right) 
        +\sum_{j=1}^D\left(-\frac{\theta}{2\rho}
          C_{\theta,\alpha}^{(j)}\right)\pd{\rho}{x_j}\\
&\qquad +\sum_{j=1}^D\sum_{d=1}^D\pd{u_d}{x_j}\left(
    \theta f_{\alpha-e_d-e_j} + (\alpha_j+1)f_{\alpha-e_d+e_j}
        -\frac{C_{\alpha}}{D\rho}p_{e_j+e_d}\right)\\
&\qquad +\sum_{j=1}^D\sum_{d=1}^D\left(
        \left(-\frac{f_{\alpha-e_d}}{\rho}\right)\pd{p_{e_j+e_d}}{x_j}
        +\frac{C_{\theta,\alpha}^{(j)}}{2D\rho}\pd{p_{2e_d}}{x_j}\right)
        +\left(-\frac{C_{\alpha}}{D\rho}\right)\sum_{j=1}^D\pd{q_j}{x_j}
=0,
\end{split}
\end{equation}
where $C_{\alpha}$ and $C_{\theta,\alpha}^{(j)}$ are defined as
\begin{subequations}\label{eq:c_alpha}
\begin{align}
C_{\alpha}  &= \sum_{k=1}^Df_{\alpha-2e_k},\\
C_{\theta,\alpha}^{(j)} &= \sum_{k=1}^D\left(\theta
            f_{\alpha-2e_k-e_j}+(\alpha_j+1)f_{\alpha-2e_k+e_j}\right).
\end{align}
\end{subequations}
Then collecting \eqref{eq:conservation_laws}, \eqref{eq:momentseqs2}
and \eqref{eq:basic_para_pij}, we get
\begin{equation}\label{eq:moment_pi}
\begin{split}
& \pd{p_{2e_i}/2}{t} + \sum_{j=1}^D u_j\pd{p_{2e_i}/2}{x_j}
  +\sum_{j=1}^D \left( \frac{1}{2}+\delta_{ij} \right)
  \rho\theta\pd{u_j}{x_j} \\
& \quad + \sum_{j=1}^D\sum_{d=1}^D
  (2\delta_{ij}+1)f_{2e_i-e_d+e_j}\pd{u_d}{x_j} 
  +\sum_{j=1}^D(2\delta_{ij}+1)\pd{f_{2e_i+e_j}}{x_j} = 0,
  \quad i = 1,\cdots,D.
\end{split}
\end{equation}
The \eqref{eq:conservation_laws} together with \eqref{eq:moment_pi}
and \eqref{eq:momentseqs2} form a moment system with infinite number
of equations, which is equivalent to \eqref{eq:momentseqs1}.


\section{System in 1D Spatial Space} \label{sec:hyperbolicms} 

In order to derive the regularization term to achieve the
hyperbolicity of the moment systems as in Section
\ref{sec:momentsystem}, we first consider in this section the special
case with homogeneous dependence of the distribution function on
spatial coordinate $\bx$ except for $x_1$ direction. Since the
velocity space is multi-dimensional, the result in this section is
essential different from \cite{Fan}. The general case in
multi-dimensional spatial space is studied in the next section based
on the results herein and the Galilean invariance of the
regularization.

In 1D spatial space, the distribution function $f(t, x_1, \bxi)$
satisfies
\begin{equation}\label{eq:1Dboltzmann}
\pd{f}{t}+\xi_1\pd{f}{x_1}=0, 
 \qquad t \in \bbR^+, \quad x_1\in\bbR,\quad \bxi \in \bbR^D.
\end{equation}
The moment system in Section \ref{sec:momentsystem} degenerates to a
simpler form. The conservation of mass, momentum and energy
\eqref{eq:conservation_laws} turn into
\begin{subequations}\label{eq:conservation_laws1D}
\begin{align}
&\pd{\rho}{t} + u_1\pd{\rho}{x_1} + 
  \rho \pd{u_1}{x_1} = 0, \\
\label{eq:u1d}
& \rho \pd{u_i}{t} + \rho u_1 \pd{u_i}{x_1} +
  \pd{p_{e_1+e_i}}{x_1} = 0, \quad i=1,\cdots,D,\\
\label{eq:theta1d}
&\frac{D}{2}\rho \pd{\theta}{t}
    +  \frac{D}{2} \rho u_1 \pd{\theta}{x_1}
  + \pd{q_1}{x_1} 
  + \sum_{i=1}^Dp_{e_1+e_i} \pd{u_i}{x_1} = 0.
\end{align}
\end{subequations}
The moment equations \eqref{eq:momentseqs2} become
\begin{equation}\label{eq:momentseqs1d}
\begin{split}
&\pd{f_{\alpha}}{t}+
    \theta\pd{f_{\alpha-e_1}}{x_1}+u_1\pd{f_{\alpha}}{x_1}
    +(\alpha_1+1)\pd{f_{\alpha+e_1}}{x_1} 
        -\frac{\theta}{2\rho}
          C_{\theta,\alpha}^{(1)}\pd{\rho}{x_1}\\
&\qquad +\sum_{d=1}^D\pd{u_d}{x_1}\left(
    \theta f_{\alpha-e_d-e_1} + (\alpha_1+1)f_{\alpha-e_d+e_1}
        -\frac{C_{\alpha}}{D\rho}p_{e_1+e_d}\right)\\
&\qquad +\sum_{d=1}^D\left(
        -\frac{f_{\alpha-e_d}}{\rho}\pd{p_{e_1+e_d}}{x_1}
        +\frac{C_{\theta,\alpha}^{(1)}}{2D\rho}\pd{p_{2e_d}}{x_1}\right)
        -\frac{C_{\alpha}}{D\rho}\pd{q_1}{x_1}
=0,
\end{split}
\end{equation}
where $C_{\alpha}$ and $C_{\theta,\alpha}^{(1)}$ are defined in
\eqref{eq:c_alpha}. The governing equations of $p_{2e_i}$
\eqref{eq:moment_pi} turn into: for $i =1,\cdots,D$,
\begin{equation}\label{eq:moment_pi1d}
\begin{split}
\pd{p_{2e_i}/2}{t} & +  u_1\pd{p_{2e_i}/2}{x_1}
+(\frac{1}{2}+\delta_{i1})\rho\theta\pd{u_1}{x_1} \\
&+ \sum_{d=1}^D(2\delta_{i1}+1)f_{2e_i-e_d+e_1}\pd{u_d}{x_1} 
+(2\delta_{i1}+1)\pd{f_{2e_i+e_1}}{x_1} = 0.
\end{split}
\end{equation}

Analogously to the moment system in Section
\ref{sec:boltzmann_equations}, let $f_{\alpha}=0$, $|\alpha|>M$ for
$M\ge 3$, and then \eqref{eq:conservation_laws1D} and
\eqref{eq:momentseqs1d}, together with \eqref{eq:moment_pi1d} form a
closed moment system corresponding to \eqref{eq:1Dboltzmann}.

\newcommand\seq[2]{#1\!:\!#2}
To facilitate the reading below in studying the moment system, some
notations are introduced as follows:
\begin{subequations}
\begin{align}
\text{if }\boldsymbol{a}=(a_1,\cdots,a_n)\in\bbR^{n},\text{ then}
    &~~\boldsymbol{a}(\seq{i}{j})=(a_i,\cdots,a_j), \\
\text{if }~\bA=(a_{ij})_{n\times n}\in\bbR^{n\times n},\text{ then}
    &~~\bA(i,\,\seq{j}{k})=(a_{i,j},\cdots,a_{i,k}), \\
    &~~\bA(\seq{i}{l},\,\seq{j}{k})=\begin{pmatrix}
            a_{i,j}& a_{i,j+1} & \cdots& a_{i,k}\\
            a_{i+1,j}& a_{i+1,j+1} & \cdots& a_{i+1,k}\\
            \vdots &\vdots&\ddots&\vdots\\
            a_{l,j}& a_{l,j+1} & \cdots& a_{l,k}
            \end{pmatrix},
\end{align}
\end{subequations}
\begin{subequations}
\begin{align}
\mathcal{D}&=\{1,2,\cdots,D\},   &
    \tilde{\alpha} &= (0, \alpha_2,\cdots,\alpha_D),     \\
\hat{\alpha} & = (\alpha_2,\cdots,\alpha_D)\in\bbN^{D-1}, &
    \alpha! &= \prod_{i=1}^D\alpha_i!, ~~~|\alpha| = \sum_{i=1}^D
    \alpha_i,\label{eq:def_hatalpha} \\
\cS_{D,M} &= \{\alpha\in\bbN^D\mid |\alpha| \le M\},     &
    \cS_{D,M}(\hat{\alpha}) &= \{\beta\in\cS_{D,M}\mid \hat{\beta} =
        \hat{\alpha}\}.\label{eq:def_SDM}
\end{align}
\end{subequations}
We permute the elements of $\cS_{D,M}$ by lexicographic order. Then
for any $\alpha\in \cS_{D,M}$,
\begin{equation}\label{eq:def_mN}
\mN(\alpha) = \sum_{i=1}^D\binom{\sum_{k=D-i+1}^D\alpha_k + i-1}{i}+1
\end{equation}
holds, where $\mN(\alpha)$ is the ordinal number of $\alpha$ in
$\cS_{D,M}$, and the cardinal number of set $\cS_{D,M}$ is $N =
\mN(Me_D) = \binom{M+D}{D}$, which is total number of moments if a
truncation with $|\alpha|\le M$ is introduced. In addition, it is
clear that for each $\alpha, \beta \in \cS_{D,M}$ and $\hat{\alpha}
\neq \hat{\beta}$,
\begin{equation}
\cS_{D,M}(\hat{\alpha})\bigcap \cS_{D,M}(\hat{\beta}) =\emptyset,\quad
\cS_{D,M} = \bigcup_{\alpha\in\cS_{D,M}}\cS_{D,M}(\hat{\alpha}),
\end{equation}
simultaneously hold.

\subsection{Structure of coefficient matrix}
Similar to the 1D case, a truncation with $|\alpha|\le M, M\ge 3$ is
applied. Let $\bw \in\bbR^N$ and for each $i,j\in\mathcal{D},$ and
$i\neq j$,
\begin{subequations}\label{eq:basic_moments}
\begin{align}
w_1 &= \rho,    &   w_{\mN(e_i)} &= u_i,\\
w_{\mN(2e_i)} &= \frac{p_{2e_i}}{2},    &
    w_{\mN(e_i+e_j)} &= p_{e_i+e_j}, \\
w_{\mN(\alpha)} &= f_{\alpha}, \quad 3\le|\alpha|\le M.
\end{align}
\end{subequations}
Combining \eqref{eq:conservation_laws1D} with \eqref{eq:moment_pi1d}
and \eqref{eq:momentseqs1d}, we obtain
\begin{equation} \label{eq:1d}
\pd{\bw}{t} + \bAM\pd{\bw}{x_1} = 0,
\end{equation}
where $\bAM$ depends on \eqref{eq:conservation_laws1D},
\eqref{eq:moment_pi1d} and \eqref{eq:momentseqs1d}.

Clearly, all the matrix $\bAM$ for any $D\in\bbN^+$, $3\le M\in\bbN$
are well-defined though quite complex. Here we first give some simple
examples and conclude a few basic properties of the matrix $\bAM$.

\begin{example}\label{ex:matrix}
  If $D=2$, the ordinal number of $\alpha$ in $\cS_{D,M}$ is
  $\mN(\alpha) = \frac{(\alpha_1 + \alpha_2 + 1)(\alpha_1 +
    \alpha_2)}{2} + \alpha_2 + 1$. The permutation of $\bw$ is showed
  in Fig. \ref{fig:order_w}.  As the simplest case, the matrix $\bA_3$
  is
\begin{equation}
\bA_3=
\left[ \begin {array}{cccccccccc} 
u_1&\rho&0&0&0&0&0&0&0&0\\
\noalign{\medskip}0&u_1&0&2{\rho}^{-1}&0&0&0&0&0&0\\
\noalign{\medskip}0&0&u_1&0&{\rho}^{-1}&0&0&0&0&0\\
\noalign{\medskip}0&\frac{3}{2}p_{{1}}&0&u_1&0&0&3&0&0&0\\
\noalign{\medskip}0&2f_{{1,1}}&p_{{1}}&0&u_1&0&0&2&0&0\\
\noalign{\medskip}0&p_{{2}}&f_{{1,1}}&0&0&u_1&0&0&1&0\\
\noalign{\medskip}{\frac {\rho{\theta}^{2}-2\theta p_{{1}}}{2\rho}}&
    4f_{{3,0}}&0&\theta&0&{\frac {2f_{2,0}}{\rho}}&u_1&0&0&0\\
\noalign{\medskip}-{\frac {3\theta f_{{1,1}}}{2\rho}}&3f_{{2,1}}&
    3f_{{3,0}}&-{\frac {f_{{1,1}}}{2\rho}}&{\frac {\rho\theta-f_{2,0}}{\rho}}&
    {\frac {3f_{{1,1}}}{2\rho}}&0&u_1&0&0\\
\noalign{\medskip}-\frac{1}{2}{\theta}^{2}&2\,f_{{1,2}}&2f_{{2,1}}&
    {\frac {2f_{2,0}}{\rho}}&-{\frac {f_{{1,1}}}{\rho}}&\theta&0&0&u_1&0\\
\noalign{\medskip}-{\frac {\theta f_{{1,1}}}{2\rho}}&f_{{0,3}}&
    f_{{1,2}}&{\frac {f_{{1,1}}}{2\rho}}&{\frac {f_{2,0}}{\rho}}&
    {\frac {f_{{1,1}}}{2\rho}}&0&0&0&u_1\end {array} \right],
\end{equation}
whereas $p_1 = p_{2e_1}, f_{m,n}=f_{me_1+ne_2}$. If $M>3$, for any
$\alpha\in\bbN^2$, and $3<|\alpha|\le M$,
\begin{subequations}\label{eq:AMD2}
\begin{align}
&\bAM(\seq{1}{10}, \, \seq{1}{10}) = \bA_3,\\
&\bAM(\mN(\alpha), \mN(\alpha)) =  u_1,\\
&\bAM(\mN(\alpha), \mN(\alpha-e_1)) =  \theta, 
    \quad\text{if }\alpha_1>0, \label{eq:AMD2:3} \\
&\bAM(\mN(\alpha), \mN(\alpha+e_1)) =  \alpha_1+1, 
    \quad\text{if }|\alpha|< M,\\ 
\begin{split}
&\bAM(\mN(\alpha),\,\seq{1}{9})
=(-\frac{\theta}{2\rho}C_{\theta,\alpha}^{(1)}, ~~
    \theta
    f_{\alpha-2e_1}+(\alpha_1+1)f_{\alpha}-\frac{C_{\alpha}}{2\rho}p_{2e_1},\\
   & \qquad\qquad\qquad\qquad\qquad\theta
    f_{\alpha-e_1-e_2}+(\alpha_1+1)f_{\alpha+e_1-e_2}-\frac{C_{\alpha}}{2\rho}p_{e_1+e_2},\\
    &\qquad\qquad\qquad\qquad\qquad-2\frac{f_{\alpha-e_1}}{\rho}+\frac{C_{\theta,\alpha}^{(1)}}{2\rho},~
    -\frac{f_{\alpha-e_2}}{\rho}, ~\frac{C_{\theta,\alpha}^{(1)}}{2\rho},
    -\frac{3C_{\alpha}}{2\rho},~ 0,~ -\frac{C_{\alpha}}{2\rho}), 
\end{split} \label{eq:AMD2:5}
\end{align}
\end{subequations}
where $C_{\alpha}$ and $C_{\theta,\alpha}^{(1)}$ are defined in
equation \eqref{eq:c_alpha}. We remark that
\begin{itemize}
\item an entry $\bAM(i,j)$, if not defined above, is taken as zero;
\item for $|\alpha|=4$, some $\bAM(i,j)$ may be double defined in
  \eqref{eq:AMD2:3} and \eqref{eq:AMD2:5}, the value of which is the
  sum of the both expressions.
\end{itemize}
\begin{figure}[h]\label{fig:sp_ori}
\centering
  \includegraphics[width=0.75\textwidth]{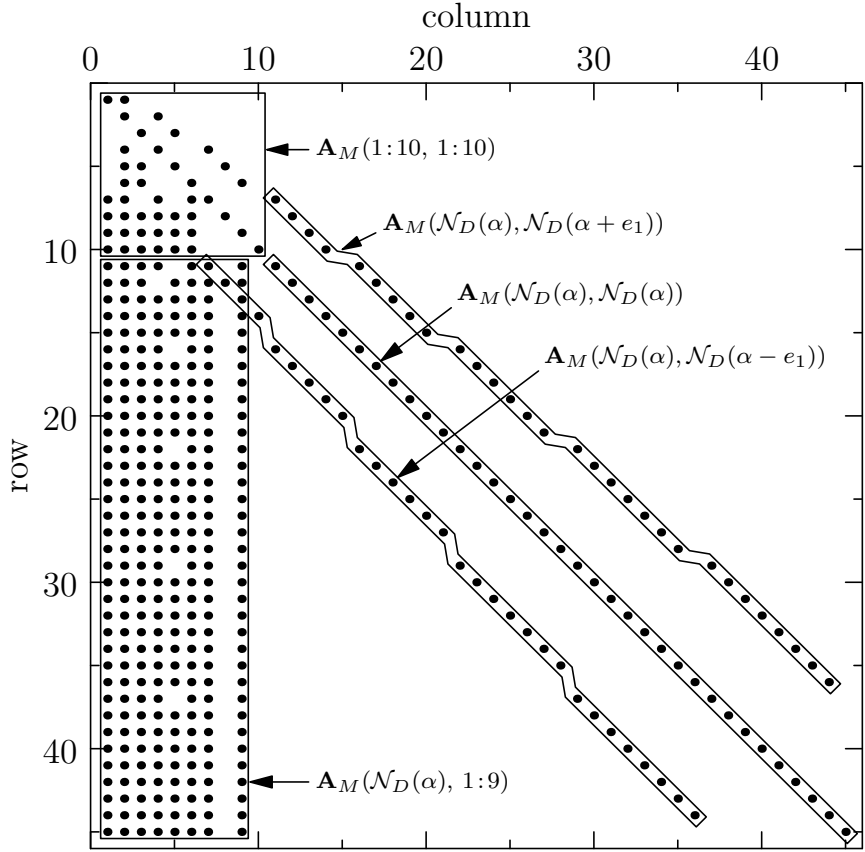}
  \caption{The sparse matrix pattern of $\bAM$ with $M=8, D=2$. Its
      nonzero entries are defined as in \eqref{eq:AMD2}.}
\end{figure}
\end{example}

Fig. \ref{fig:sp_ori} gives the sparse matrix pattern of $\bAM$ with
$M=8$ and $D=2$. It is observed that there is no more than one nonzero
component of $\bAM(i, \seq{i+1}{N})$, for each $i=1$, $\cdots$,
$N$. Precisely, there is a unique nonzero component as $1 \le i \le
\mN((M-1)e_D)$ and $\bAM(i, \seq{i+1}{N})={\bf 0}$ as
$i>\mN((M-1)e_D)$.  Noticing that the column index of the nonzero
entries in the upper triangular part of $\bAM$ on different rows are
different from each other, this makes one recall the form of lower
Hessenberg matrix, of which the only nonzero entries in the upper
triangular part on the $i$-th row is located at position $(i,
i+1)$. This property of Hessenberg matrix makes it very convenient for
one to calculate its eigenvectors once eigenvalues are given using a
row by row sequential procedure. Here $\bAM$ is essentially the same
as lower Hessenberg matrix on this point, and we notice that its lower
triangular part is sparse, hence we are provided the approach to
calculate the eigenvalues $\bAM$ together with the corresponding
eigenvectors using the same technique.

Furthermore, \eqref{eq:AMD2} shows the diagonal entries of the matrix
$\bAM$ are all $u_1$, and the entries of the matrix
$\bAM-u_1\boldsymbol{\rm I}$ are independent of $u_i,i\in\mathcal{D}$,
where $\boldsymbol{\rm I}$ is the $N\times N$ identity matrix.
In fact, \eqref{eq:1d} can be written as
\begin{equation}\label{eq:1d_material der}
\frac{{\rm D} \bw}{{\rm D} t}+\left(\bAM-u_1{\bf I}\right)
        \pd{\bw}{x_1}=0,
\end{equation}
where $\dfrac{{\rm D}~}{{\rm D}t}$ is material derivative defined as
\[
\frac{{\rm D}~}{{\rm D}t}=\pd{}{t}+u_1\pd{}{x_1}.
\]
Hence, that $\bAM-u_1\boldsymbol{\rm I}$ are independent of
$u_i,i\in\mathcal{D}$ indicates the moment system is translation
invariant. On the other hand, eigenvalues of $\bAM$ can be written in
the form $u_1+a$, where $a$ is indeterminate and independent of
$u_i,~i\in\mathcal{D}$, and eigenvectors of it is independent of
$u_i,~i\in\mathcal{D}$, too.

\begin{figure}[h]
\centering
\subfigure[The permutation of $\bw$]{
  \label{fig:order_w}
  \includegraphics[width=0.45\textwidth]{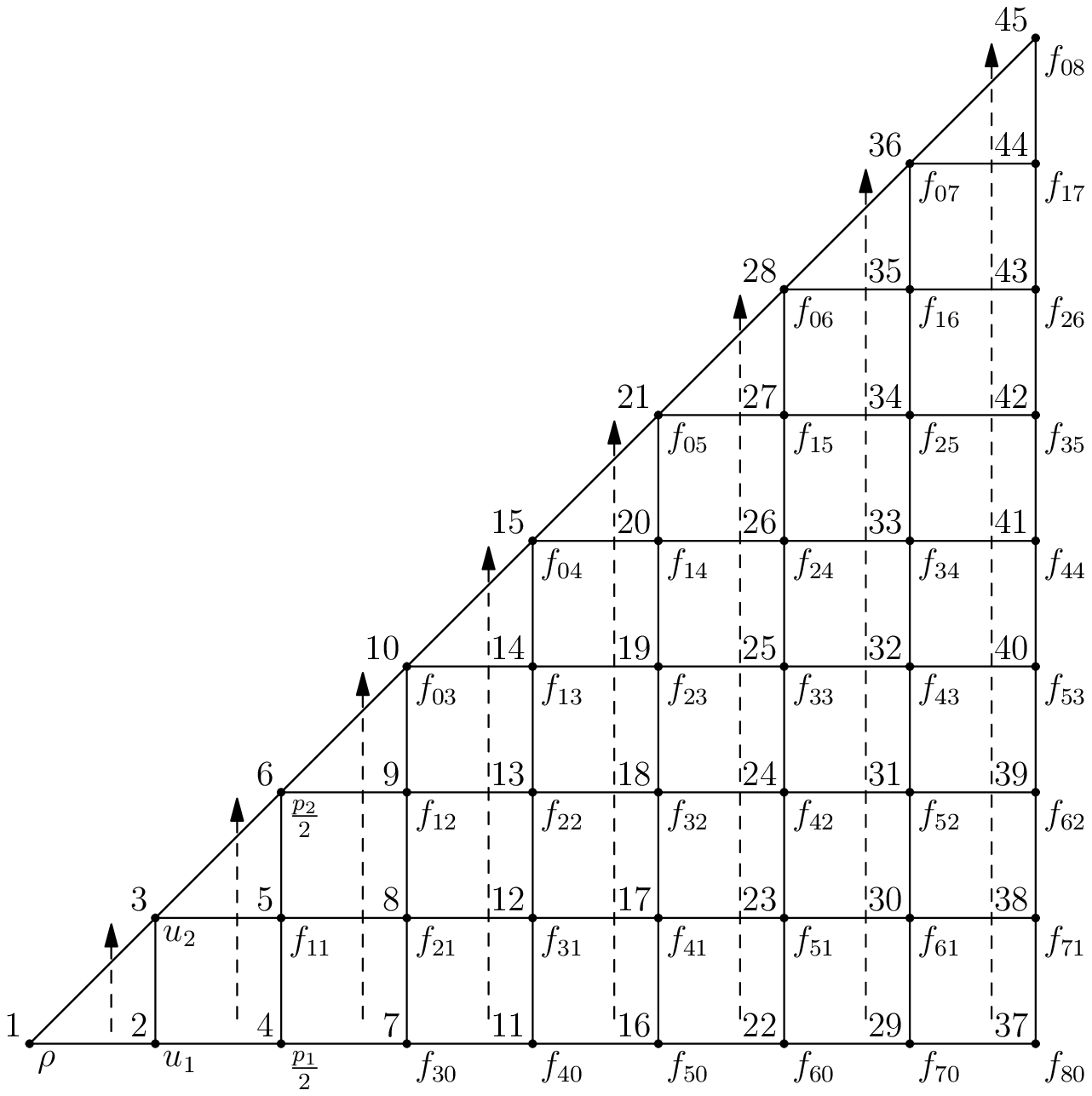}
}
\subfigure[A permutation of $\bw'$ defined in example
\ref{ex:permutation}]{
  \label{fig:order_w2}
  \includegraphics[width=0.45\textwidth]{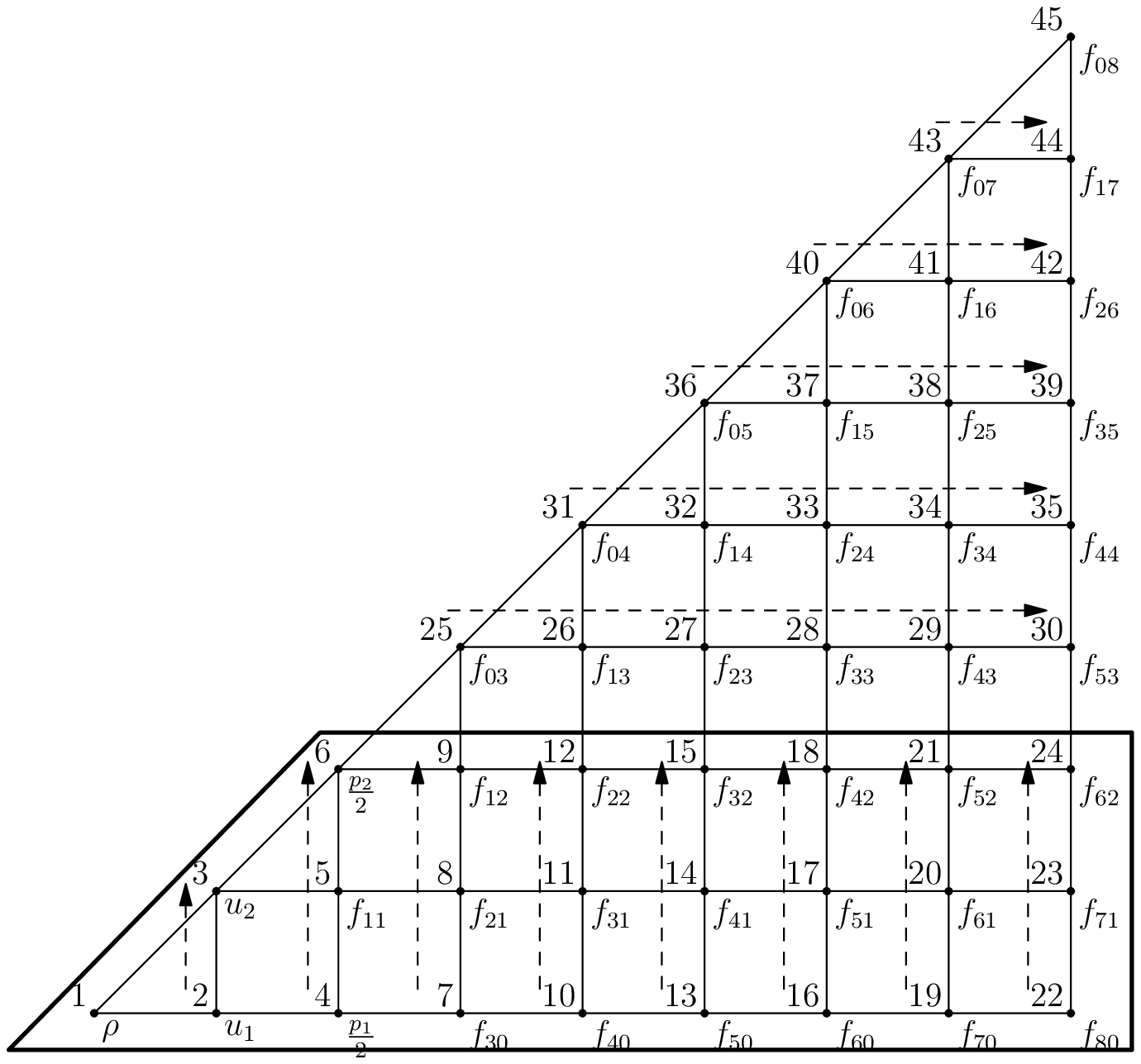}
}
\caption{The permutation of moments while $D=2, M=8$. Each node stands
  for one moment. The marks in the lower right of the node shows the
  expression of the moment, while the number in the upper left
  represents the ordinal number in $\bw$ or $\bw'$. The dashed arrows
  depict the path of the corresponding permutation. The left one is
  the permutation of $\bw$, and the right one is a permutation of
  $\bw'$ defined in example \ref{ex:permutation}.}
\end{figure}

\begin{example}
  Considering the case $D=2$, we can write out the matrix $\bAM$
  according to example \ref{ex:matrix}, for any $3\le M\in\bbN$. If we
  let $f_{\alpha}=0$ for all $\alpha\in\bbN^2$, and $|\alpha|\le M$
  except $f_0$ and $f_{Me_1}$, direct calculation gives the
  characteristic polynomial of $\bAM$ as
\begin{equation}
\begin{split}
&|\lambda {\bf I}-\bAM|= \left(\prod_{i=1}^{M-1}\He_i\left(\frac{\lambda-u_1}{\sqrt{\theta}}\right)\theta^{i/2}\right)\\
&\qquad ~ \times \left(\He_M\left(\frac{\lambda-u_1}{\sqrt{\theta}}\right)\theta^{M/2}+(-1)^{M-1}M!f_{Me_1}\right)\\
&\qquad ~ \times \left(\He_{M+1}\left(\frac{\lambda-u_1}{\sqrt{\theta}}\right)\theta^{(M+1)/2}+(-1)^{M-1}(M+1)!f_{Me_1}(\lambda-u_1)\right).
\end{split}
\end{equation}
The matrix $\bAM$, obviously, has complex eigenvalues, for some
$f_{Me_1}$.
\end{example}

Analogously, involved calculations with the help of computer algebraic
system show that the matrix $\bAM$ has complex eigenvalues for some
admissible $\bw$, for any $D \geq 3$. This reveals that $\bAM$ is not
diagonalisable with real eigenvalues for some $\bw$ as $D=1$ in
\cite{Fan}.
\begin{remark}
  Since moments $f_{\alpha}$ are related to $f(t,\bx,\bxi)$ by
  \eqref{eq:expansion}, the moments $f_{\alpha}$ can not be arbitrary
  number if the distribution function is to be kept
  positive. Particularly, $\rho$ and $\theta$ given by
  \eqref{eq:def_rhoutheta} clearly satisfy
  \begin{equation}\label{eq:def_admissible}
    \rho>0,\quad \theta>0.
  \end{equation}
  Though \eqref{eq:def_admissible} is not enough to provide us a
  positive $f(t,\bx,\bxi)$, the discussion in this paper requires no
  further constraints on the other moments. Hence, in this paper the
  admissible $\bw$ stands for $\bw$ which satisfies
  \eqref{eq:def_admissible}.
\end{remark}

\begin{example}\label{ex:permutation}
  Actually, it can be observed that the matrix $\bAM$ is reducible if
  we rearrange $\bw$ as $\bw'$ using another permutation rule. In case
  of $D=2$, the rule reads:
  \begin{enumerate}
  \item The moments with $\alpha_2 \leq 2$ are arranged at first using
    the lexicographic order;
  \item The rest moments are arranged then using the lexicographic order
    based on index transformed as $(\alpha_2, \alpha_1)$.
  \end{enumerate}
  Clearly, $\bw$ and $\bw'$ are related by a permutation matrix ${\bf
      P}$
  that $\bw'={\bf P}\bw$. The Fig. \ref{fig:order_w2} gives a schematic
  diagram of the permutation rule for $\bw'$ with $M=8$.  Let
  ${\bf A}_M'={\bf P}\bAM {\bf P}^{-1}$, then
\begin{equation}
\pd{\bw'}{t}+{\bf A}_M'\pd{\bw'}{x_1}=0,
\end{equation}
holds. Fig. \ref{fig:sp_pmt} gives the sparse matrix pattern of
${\bf A}_M'$. It is obvious that ${\bf A}_M'$ is reducible (see
e.g. \cite{Demmel} for the definition), and can be reduced into $M-1$
blocks. $\cS_{D,M}(\hat{e}_1)$ $\bigcup$ $\cS_{D,M}(\hat{e}_2)$ $\bigcup$
$\cS_{D,M}(2\hat{e}_2)$ is one of the blocks, and
$\cS_M^D(\hat{\alpha})$, for each $\hat{\alpha}\in\bbN$ and
$\hat{\alpha}\neq \hat{e}_1, \hat{e}_2, 2\hat{e}_2$ is another block.
\end{example}

\begin{figure}[h]\label{fig:sp_pmt}
\centering
  \includegraphics[width=0.65\textwidth]{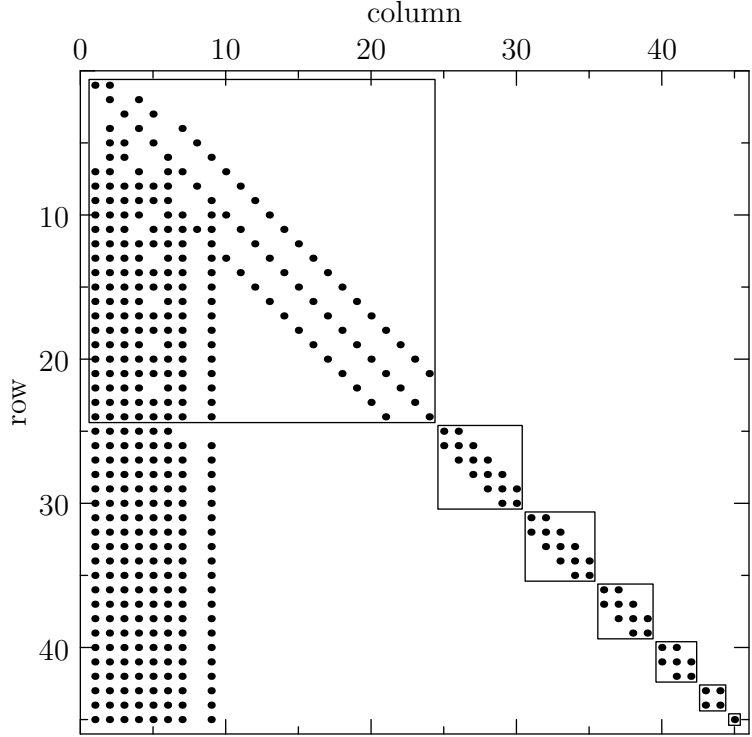}
  \caption{The sparse matrix pattern of ${\bf A}_M'$ with $M=8, D=2$. 
      ${\bf A}_M'$ is reducible.}
\end{figure}

These examples show some very useful properties of the matrix $\bAM$
as follows:
\begin{property}\label{pro:properties_A}
Matrix $\bAM$ satisfies the following properties:
\renewcommand\theenumi{\ref{pro:properties_A}(\arabic{enumi})}
\makeatother
\begin{enumerate}
\item\label{enuit:nozero} 
  For each $\alpha\in\bbN^D$, $|\alpha|\le M$, let $i=\mN(\alpha)$,
  then there are no more than one entries of $\bAM(i,\,\seq{i+1}{N})$
  nonzero. In fact, there is a unique nonzero component as
  $|\alpha|<M$ and $\bAM(i,\,\seq{i+1}{N})=0$ as
  $|\alpha|=M$.
\item\label{enuit:diag} 
  The diagonal entries of the matrix $\bAM$ are all $u_1$, and entries
  of the matrix $\bAM- u_1{\boldsymbol{\rm I}}$ are independent of
  $u_i,i\in\mathcal{D}$.
\item\label{enuit:diagonalisable} 
  $\bAM(\bw)$ may be not diagonalisable with real eigenvalues for some
  admissible $\bw$.
\item\label{enuit:reducible} 
  $\bAM$ is reducible, and can be reduced into $
  \binom{M+D-1}{D-1}-2(D-1)$ blocks.  $\cS_{D,M}(\hat{e}_1)$ $\bigcup$
  $\left(\bigcup\limits_{k=2}^D\cS_{D,M}(\hat{e}_k)\right)$ $\bigcup$
  $\left(\bigcup\limits_{k=2}^D\cS_{D,M}(2\hat{e}_k)\right)$ is one of
  the blocks, and $\cS_{D,M}(\hat{\alpha})$, for each
  $\hat{\alpha}\in\bbN^{D-1}$ and $\hat{\alpha}\neq \hat{e}_1,
  \hat{e}_k, 2\hat{e}_k$, $k=2,\cdots,D$ is one of the blocks.
\end{enumerate}
\end{property}

\subsection{Globally hyperbolic regularization}
For 1D case, the regularization as \eqref{eqRM1D} was proposed such
that the moment system turns out to be globally hyperbolic (see
\cite{Fan} for details). Actually, the regularization therein can be
extend to multiple dimensional systems. For $D\in\bbN^+$, let us start
from the definition as below:
\begin{definition}
For any $|\alpha|=M$, let
\begin{equation}\label{defRM}
\RM(\alpha) \triangleq \sum_{j=1}^D \RM^j(\alpha),
\end{equation}
where
\begin{equation}\label{defRMj}
\RM^j(\alpha) =
    \sum_{d=1}^Df_{\alpha-e_d+e_j}\pd{u_d}{x_j} + 
    \frac{1}{2}\left(\sum_{d=1}^Df_{\alpha-2e_d+e_j}\right)\pd{\theta}{x_j}.
\end{equation}
As in 1D case \cite{Fan}, $\RM(\alpha)$ is the regularization terms
based on the characteristic speed correction.
\end{definition}
With \eqref{eq:pdtheta}, we have that
\[
\RM^j(\alpha) = \sum_{d=1}^Df_{\alpha-e_d+e_j}\pd{u_d}{x_j} + 
\left(\sum_{d=1}^Df_{\alpha-2e_d+e_j}\right)
\left(\sum_{i=1}^D\frac{1}{D\rho}\frac{\partial p_{2e_i}/2}{\partial
    x_j} -\frac{\theta}{2\rho}\frac{\partial \rho}{\partial
    x_j}\right).
\]
For the case that the dependence of $f$ on $\bx$ is only on $x_1$, we
have that
\[ \RM^j(\alpha) = 0,\quad {\rm for~~} j = 2, \cdots, D. \] This leads
to $\RM(\alpha) = \RM^1(\alpha)$. The regularized system is obtained
by subtracting $\RM(\alpha)$ from the governing equation of $f_\alpha$
in \eqref{eq:1d}, for $|\alpha| = M$.
\begin{definition} \label{def:bhAM}
$\bhAM$ is called the regularized matrix
of the matrix $\bAM$,  if it satisfies that for any admissible
$\bw$, 
\begin{equation}\label{eq:definition_bhAM}
\bhAM\pd{\bw}{x_1} = \bAM\pd{\bw}{x_1} - \sum_{|\alpha|=M}\RM^1(\alpha)I_{\mN(\alpha)},
\end{equation}
where $I_k$ is the $k$-th column of the $N \times N$ identity matrix.
\end{definition}
The regularization terms only change a few entries of the lower
triangular part of $\bAM$, with the order of the corresponding moments
equals to $M$, so that the Properties \ref{enuit:nozero},
\ref{enuit:diag}, \ref{enuit:reducible} of $\bAM$ are also valid for
$\bhAM$, while the Property \ref{enuit:diagonalisable} is changed to
be the diagonalisability of $\bhAM$ over the real field $\bbR$.
Actually, we have the following theorem:
\begin{theorem}\label{thm:diagonalisable}
The regularized moment system 
\begin{equation}\label{eq:hyperbolicsystem1d}
\pd{\bw}{t} + \bhAM\pd{\bw}{x_1}=0
\end{equation}
is hyperbolic for any admissible $\bw$.
\end{theorem}
The definition of the hyperbolicity shows that this theorem is
equivalent to the diagonalisability of $\bhAM$ with real eigenvalues
for any admissible $\bw$. Before proving this result, we first make
some simplifications and give several useful lemmas.

Let us denote that
\begin{subequations}\label{eq:def_Lambda}
\begin{align}
\boldsymbol{d} &= (d_j)_{N\times 1}, & & d_1 = \rho^{-1},\\
d_{j+1} &= \theta^{-1/2},\quad j=\mathcal{D}, & & d_{\mN(\alpha)} =
    \rho^{-1}\theta^{-|\alpha|/2}, \quad 2\le|\alpha|\le M,\\
\boldsymbol{\Lambda} &= \mathrm{diag}\left\{d_1,d_2,\cdots,
    d_N\right\},\\
\hat{p}_{e_i+e_k} &= \frac{p_{e_i+e_k}}{\rho\theta},\quad i,k=\mathcal{D},& 
    &g_{\alpha} = \frac{f_{\alpha}}{\rho\theta^{|\alpha|}}.
\end{align}
\end{subequations}
By the virtue of Property \ref{enuit:diag}, we let
\begin{equation}\label{eq:dimensionless}
\bhAM = u_1\boldsymbol{\mathrm{I}} +
    \sqrt{\theta}\boldsymbol{\Lambda}^{-1}\btAM\boldsymbol{\Lambda}.
\end{equation}
Then, we can obtain properties of $\bhAM$ by studying the matrix
$\btAM$. Since $\btAM$ is related to $\bhAM - u_1 \boldsymbol{\rm{I}}$
by a similarity transformation by a diagonal matrix, Property
\ref{enuit:nozero} holds for $\btAM$. Hence, it is convenient to
calculate the eigenvectors of matrix $\btAM$.
Firstly, we denote by some symbols as
\begin{align}
\label{eq:def_v}
r_{w_{\mN(\alpha)}} & = v_{\mathcal{N}_{D-1}(\hat{\alpha})},  &&\text{if $\alpha_1 = 0$},  \\
r_{u_1} & = \lambda r_{\rho},       & 
    &r_{p_{2e_1}/2}  = \frac{\lambda^2}{2}r_{\rho}, \\
r_{f_{e_i}}&=r_{f_{\alpha}}=0,   &   & \text{if at least one
    $\alpha_j<0, j\in\mathcal{D}$}\label{eq:rf0}\\ 
r_{p_{e_1+e_k}} & = \lambda r_{u_k},  && k
\in\mathcal{D}\backslash\{1\},\label{eq:bre1k}
\end{align}
\begin{align}
r_{f_{2e_i}}&=r_{p_{2e_i}/2}-\sum_{d=1}^D\frac{r_{p_{2e_d}/2}}{D},\\
r_{w_{\mN(\alpha)}} &= \frac{\He_{\alpha_1}(\lambda)}{\alpha_1!}
    \big(r_{f_{\tal}}+G(\tal)\big) - G(\alpha),\quad 
    \text{if } \alpha_1\neq 0,~ |\alpha|\ge 3.
\end{align}
where $v_{\mathcal{N}_{D-1}(\hat{\alpha})}$ and $\lambda$ are
indeterminate parameters, and
\begin{equation}\label{eq:G}
G(\alpha) = \sum_{d=1}^D g_{\alpha-e_d}r_{u_d} +
    \left(\sum_{d=1}^D\frac{r_{p_{2e_d}/2}}{D}- \frac{r_{\rho}}{2}\right)\sum_{k=1}^D
    g_{\alpha-2e_k}.
\end{equation}
For better readability, here we adopt the notations [e.g. $w_k$,
$f_{e_i},~f_{2e_i}$] as the subscript of $r$. This does not mean the
subscripts are taken as the value of them, but only taken literally as
the notations themselves. We collect these $r_{w_k}$, with $w_k$ as a
component of $\bw$ and $k=1$, $\cdots$, $N$, to produce a vector 
$\br\in\bbR^N$ as
\begin{equation}\label{eq:def_br}
\br = (r_{w_1}, r_{w_2}, \cdots, r_{w_N}),
\end{equation}
where $N$ is total number of moments. Here it is clear that $\br$ is
prescribed once $v_{\mathcal{N}_{D-1}(\hat{\alpha})}$ and $\lambda$
are all given. With particular setup of these parameters, $\lambda$
and $\br$ is turned out to be a pair of eigenvalue and eigenvector of
$\btAM$. Precisely, we have the following lemma:
\begin{lemma}\label{lem:eigenvectors}
  $\br \neq 0$ is the right eigenvector of the matrix
  $\btAM$ for the eigenvalue $\lambda$ if
\begin{equation}\label{eq:eigenvectorscond}
  \frac{\He_{\alpha_1+1}(\lambda)}{(\alpha_1+1)!}
  \big(r_{w_{\tal}}+G(\tal)\big)=0
\end{equation}
holds, for all $|\alpha|=M$.
\end{lemma}
\begin{proof}
Let $i=\mN(\alpha)$, with $|\alpha|\le M$, then we need only to
verify
\begin{equation}\label{eq:eig_basis}
\btAM(i, 1:N)\cdot\br=\lambda r_{w_i}
\end{equation}
always valid. Since $\bAM$ is determined by
\eqref{eq:conservation_laws1D}, \eqref{eq:momentseqs1d} and
\eqref{eq:moment_pi1d}, and $\bhAM$ and $\btAM$ are defined as
\eqref{eq:definition_bhAM} and \eqref{eq:dimensionless}, respectively,
we can write any entries of $\btAM$. Now let us verify the equation
\eqref{eq:eig_basis} case by case:
\begin{itemize}
\item{For $\alpha=0$},
\begin{equation}\label{eq:eigalpha0}
\btAM(i,1:N)\cdot\br=1\cdot r_{u_1}=\lambda r_{\rho}.
\end{equation}
\item{For $\alpha=e_1$},
\begin{equation}\label{eq:eigalphae1}
\btAM(i,1:N)\cdot\br=2\cdot r_{p_{2e_1}/2}=\lambda^2r_{\rho}=\lambda
    r_{u_1}.
\end{equation}
\item{For $\alpha=e_k$, $k=2$, $\cdots$, $D$,}
\begin{equation}\label{eq:eigalphaek}
\btAM(i,1:N)\cdot\br=1\cdot r_{p_{e_1+e_k}/2}=\lambda r_{u_k}.
\end{equation}
\item{For $\alpha=2e_1$},
\begin{equation}\label{eq:eigalpha2e1}
\begin{split}
\btAM(i,1:N)\cdot\br&=\frac{3}{2}r_{u_1}+\sum_{d=1}^D3g_{\alpha+e_1-e_d}r_{u_d}+3r_{f_{3e_1}}\\
                    &=\frac{3}{2}r_{u_1}+3g_{2e_1}r_{u_1}+
                        3\left(\frac{\He_3(\lambda)}{6}r_{\rho}-g_{2e_1}r_{u_1}\right)\\
                    &=\frac{\lambda^3}{2}r_{\rho}=\lambda
                    r_{p_{2e_1}/2}.
\end{split}
\end{equation}\label{eq:eigalpha2ek}
\item{For $\alpha=2e_k$, $k=2,\cdots,D$}
\begin{equation}
\begin{split}
&~~~~\btAM(i,1:N)\cdot\br \\
&=\frac{1}{2}r_{u_1}+\sum_{d=1}^Dg_{\alpha+e_1-e_d}r_{u_d}+r_{f_{e_1+2e_k}}\\
&=\frac{1}{2}r_{u_1}+g_{\alpha}r_{u_1}+g_{e_1+e_k}r_{u_k}+
  \left(\lambda\big(r_{p_{2e_k}/2}-\frac{1}{2}r_{\rho}\big)-g_{\alpha}r_{u_1}-g_{e_1+e_k}r_{u_k}\right)\\
&=\lambda r_{p_{2e_k}/2}.
\end{split}
\end{equation}
\item{For $\alpha=e_1+e_k$, $k=2,\cdots,D$}
\begin{equation}\label{eq:eigalphae1k}
\begin{split}
&~~~~\btAM(i,1:N)\cdot\br\\
&=\frac{1}{2}r_{u_1}+\sum_{d=1}^Dg_{\alpha+e_1-e_d}r_{u_d}+r_{f_{e_1+2e_k}}\\
&=\frac{1}{2}r_{u_1}+g_{\alpha}r_{u_1}+g_{e_1+e_k}r_{u_k}+
  \left(\lambda\big(r_{p_{2e_k}/2}-\frac{1}{2}r_{\rho}\big)-g_{\alpha}r_{u_1}-g_{e_1+e_k}r_{u_k}\right)\\
&=\lambda r_{p_{2e_k}/2}.
\end{split}
\end{equation}
\item{For $3\le|\alpha|<M$, $\alpha_1>0$}
\begin{equation}\label{eq:eigalphaelse}
\begin{split}
\btAM(i,1:N)\cdot\br&=
    1\cdot r_{f_{\alpha-e_1}} +(\alpha_1+1)r_{f_{\alpha+e_1}}
    -\frac{1}{2} \tilde{C}_{\theta,\alpha}^{(1)}r_{\rho}\\
&\qquad +\sum_{d=1}^D\left(
    g_{\alpha-e_d-e_1} + (\alpha_1+1)g_{\alpha-e_d+e_1}
        -\frac{\tilde{C}_{\alpha}}{D}\hat{p}_{e_1+e_d}\right)r_{u_d}\\
&\qquad +\sum_{d=1}^D\left(
        -g_{\alpha-e_d}r_{p_{e_1+e_d}}
        +\frac{\tilde{C}_{\theta,\alpha}^{(1)}}{D}r_{p_{2e_d}/2}\right)
        -\frac{\tilde{C}_{\alpha}}{D}r_{q_1}\\
&\triangleq X_1+X_2,
\end{split}
\end{equation}
whereas
\begin{align*}
\tilde{C}_{\alpha}&=\sum_{k=1}^Dg_{\alpha-2e_k},\\
\tilde{C}_{\theta,\alpha}^{(1)}&=\sum_{k=1}^D\left(g_{\alpha-2e_k-e_1}+(\alpha_1+1)g_{\alpha-2e_k+e_1}\right),\\
r_{q_1}&=3r_{f_{3e_1}}+\sum_{k=2}^Dr_{f_{e_1+2e_k}}
    =-\frac{D}{2}\lambda r_{\rho}+\sum_{k=1}^D\left(\lambda
                r_{p_{2e_k}/2}-\hat{p}_{e_1+e_k}r_{u_k}\right),\\
X_1&= r_{f_{\alpha-e_1}} +(\alpha_1+1)r_{f_{\alpha+e_1}},\quad X_2
\text{ is the rest terms.}
\end{align*}
Substituting \eqref{eq:G} into $X_1$ yields
\begin{equation}\label{eq:eigalpha_x1}
\begin{split}
X_1&=\frac{\He_{\alpha_1-1}(\lambda)}{(\alpha_1-1)!}\big(r_{f_{\tal}}+G(\tal)\big)-G(\alpha-e_1) \\
&~~~+(\alpha_1+1)\left(\frac{\He_{\alpha_1+1}(\lambda)}{(\alpha_1+1)!}\big(r_{f_{\tal}}+G(\tal)\big)-G(\alpha+e_1)\right)\\
   &=\lambda\frac{\He_{\alpha_1}(\lambda)}{\alpha_1!}\big(r_{f_{\tal}}+G(\tal)\big)
    -G(\alpha-e_1)-(\alpha_1+1)G(\alpha+e_1).
\end{split}
\end{equation}
For $X_2$, using \eqref{eq:bre1k}, we get
\begin{equation}\label{eq:eigalpha_x2}
\begin{split}
X_2&=
    -\frac{r_{\rho}}{2}\sum_{k=1}^D\left(g_{\alpha-e_1-2e_k}+(\alpha_1+1)g_{\alpha+e_1-2e_k}-\lambda
                        g_{\alpha-2e_k}\right)\\
    &\qquad+\sum_{d=1}^Dr_{u_d}\left(g_{\alpha-e_1-e_d}+(\alpha_1+1)g_{\alpha+e_1-e_d}-\lambda g_{\alpha-e_d}\right)\\
    &\qquad+\left(\sum_{d=1}^D\frac{r_{p_{2e_d}/2}}{D}\right)
        \sum_{k=1}^D\left(g_{\alpha-e_1-2e_k}+(\alpha_1+1)g_{\alpha+e_1-2e_k}-\lambda
                        g_{\alpha-2e_k}\right).\\
\end{split}
\end{equation}
Now we calculate $G(\alpha-e_1)+(\alpha_1+1)G(\alpha+e_1)-\lambda
G(\alpha)$. Some simplification gives
\begin{equation}\label{eq:eigalpha_x2_2}
\begin{split}
&~~~~G(\alpha-e_1)+(\alpha_1+1)G(\alpha+e_1)-\lambda G(\alpha)\\
    &=\sum_{d=1}^Dr_{u_d}\left(g_{\alpha-e_1-e_d}+(\alpha_1+1)g_{\alpha+e_1-e_d}-\lambda g_{\alpha-e_d}\right)\\
    &~~~~+\left(\sum_{d=1}^D\frac{r_{p_{2e_d}/2}}{D}-\frac{r_{\rho}}{2}\right)
        \sum_{k=1}^D\left(g_{\alpha-e_1-2e_k}+(\alpha_1+1)g_{\alpha+e_1-2e_k}-\lambda
                        g_{\alpha-2e_k}\right)\\
&=X_2.
\end{split}
\end{equation}
\eqref{eq:eigalpha_x1}, \eqref{eq:eigalpha_x2} and
\eqref{eq:eigalpha_x2_2} show
\begin{equation}
X_1+X_2=\lambda\frac{\He_{\alpha_1}(\lambda)}{\alpha_1!}\big(r_{f_{\tal}}+G(\tal)\big)
        -\lambda G(\alpha)=\lambda r_{w_i}.
\end{equation}

\item For $3\le|\alpha|<M$, $\alpha_1=0$ or $\alpha=e_k+e_j$, $j>k>1$,
    it is the case that to let $r_{f_{\alpha-e_1}}=0$ in
    \eqref{eq:eigalphaelse}, which is actually part of \eqref{eq:rf0}.
    Hence, \eqref{eq:eig_basis} is valid in this case.
\item For $|\alpha|=M$, if $\alpha_1>0$, then
\eqref{eq:definition_bhAM} and \eqref{eq:dimensionless} show this case
equals to let
\begin{equation} \label{eq:case_M}
r_{f_{\alpha+e_1}}+
    \sum_{d=1}^Dg_{\alpha-e_d+e_1}r_{u_d} +
    \left(\sum_{d=1}^Dg_{\alpha-2e_d+e_1}\right)
    \left(\sum_{i=1}^D\frac{1}{D}r_{p_{2e_i}/2}
    -\frac{1}{2}r_{\rho}\right)=0
\end{equation}    
in \eqref{eq:eigalphaelse}.
Since 
\[
G(\alpha+e_1)=\sum_{d=1}^Dg_{\alpha-e_d+e_1}r_{u_d} +
    \left(\sum_{d=1}^Dg_{\alpha-2e_d+e_1}\right)
    \left(\sum_{i=1}^D\frac{1}{D}r_{p_{2e_i}/2}
    -\frac{1}{2}r_{\rho}\right),
\]
we need only to prove 
\[
r_{f_{\alpha+e_1}}+G(\alpha+e_1)=0.
\]
Actually, it is what \eqref{eq:eigenvectorscond} tells.

If $\alpha_1=0$, 
\eqref{eq:definition_bhAM} and \eqref{eq:dimensionless} show this case
equals to let
$r_{f_{\alpha-e_1}}=0$ and \eqref{eq:case_M} valid in
\eqref{eq:eigalphaelse}. The former is part of \eqref{eq:rf0}, while
the latter is proved above. Hence,
\eqref{eq:eig_basis} is valid in this case.
\end{itemize}
Collecting all the cases above, we conclude \eqref{eq:eig_basis} is
valid for arbitrary $\alpha$. The lemma is proved.
\end{proof}
For any $\alpha$, let $\beta=\alpha+ke_1, k\in\bbN$, then
$\tal=\tilde{\beta}$ holds. Therefore, $r_{w_{\tal}}, |\alpha|\le M$
is equivalent to $r_{w_{\tal}}, |\alpha|=M$.  Hence, in the lemma
\ref{lem:eigenvectors}, parameters $r_{\tal}, |\alpha| = M$ and
$\lambda$, are all indeterminate. Let
$N_v=\mathcal{N}_{D-1}(M\hat{e}_{D})$, and \[\bv =
(v_1,v_2,\cdots,v_{N_v})\in \bbR^{N_v}.\] Since $\br$ is determined by
$\bv$ and $\lambda$, by studying the space of the parameters $\bv$ and
$\lambda$, we can fully clarify the structure of the eigenvectors of
$\btAM$. We have the following lemma that
\begin{lemma}\label{lem:numberofeigenvectors}
$\btAM$ has $N$ linearly independent eigenvectors.
\end{lemma}
\begin{proof}
For $|\alpha|=M$, \eqref{eq:eigenvectorscond} can be written as
\begin{itemize}
\item if $\tal = 0$, then $\mathcal{N}_{D-1}(\hat{\alpha})=1$ and  
\begin{equation}
v_1\He_{M+1}(\lambda) = 0,\label{eq:eigenvectors_rho}
\end{equation}

\item if $\tal = e_k$, $k\in\mathcal{D}\backslash\{1\}$, then $\mathcal{N}_{D-1}(\hat{\alpha})=k$ and  
\begin{equation}
v_{k}\He_{M}(\lambda) = 0, \quad \label{eq:eigenvectors_u}
\end{equation}

\item if $\tal=2e_k$, $k\in\mathcal{D}\backslash\{1\}$, then
\begin{equation}
\label{eq:eigenvectors_p}
\begin{split}
0 & =\left(v_{\mathcal{N}_{D-1}(2\hat{e}_k)}+
    \sum_{d=1}^D\frac{r_{w_{\mN(2e_d)}}}{D}-\frac{v_1}{2}\right)\He_{M-1}(\lambda)\\
&=\left(v_{\mathcal{N}_{D-1}(2\hat{e}_k)}+
    \sum_{d=2}^D\frac{v_{\mathcal{N}_{D-1}(2\hat{e}_d)}}{D} + \frac{\lambda^2}{2D}v_1-\frac{v_1}{2}\right)\He_{M-1}(\lambda),
\end{split}
\end{equation}

\item if $\tal=e_k+e_l$, $k\neq l$, and
$k,l\in\mathcal{D}\backslash\{1\}$, then 
\begin{equation}\label{eq:eigenvectors_eij}
v_{\mathcal{N}_{D-1}(\hat{\alpha})} \He_{M-1}(\lambda)=0,
\end{equation}
\item otherwise ($3\le|\tal|\le M$),
\begin{equation}
\label{eq:eigenvectors_all}
\left(v_{\mathcal{N}_{D-1}(\hat{\alpha})} + 
    \sum_{d=2}^Dg_{\tal-e_d}v_{d} +
    \sum_{i=1}^Dg_{\tal-2e_i}\left(
    \sum_{d=2}^D\frac{v_{\mathcal{N}_{D-1}(2e_d)}}{D} +
        \frac{\lambda^2}{2D}v_1-\frac{v_1}{2}\right)\right)\He_{\alpha_1+1}(\lambda)=0.
\end{equation}
\end{itemize}
Let
\begin{equation}\label{eq:definition_zlambda}
\begin{array}{rcl}
\boldsymbol{z}_{\lambda} =
(\He_{M+1}(\lambda),~\cdots,&\underbrace{\He_k(\lambda),~\cdots,~\He_k(\lambda)},&
            \cdots,~\He_{1}(\lambda)),\\ [4mm]
          &\displaystyle{\scriptstyle \binom{D-1+M-k}{M+1-k}} \rm{~entries}&
\end{array}
\end{equation}
where the $\mathcal{N}_{D-1}(\hat{\alpha})$-th component of
$\boldsymbol{z}_\lambda$ is $\He_{\alpha_1+1}(\lambda)$, and
the cardinal number of set
\[ \#\{\alpha\mid|\alpha|=M,~\alpha_1=k-1\} = \binom{D-1+M-k}{M+1-k}. \]
Equations \eqref{eq:eigenvectors_rho}, \eqref{eq:eigenvectors_u},
\eqref{eq:eigenvectors_p}, \eqref{eq:eigenvectors_eij} and
\eqref{eq:eigenvectors_all} can be collected as
\begin{equation}\label{eq:zbv}
\boldsymbol{z}_{\lambda}\circ\bB\bv=0,
\end{equation}
where $\boldsymbol{c}=\boldsymbol{a}\circ\boldsymbol{b}$, stands for
$c_i = a_i b_i$, $i=1$, $\cdots$, $n$, and $\bB$ is a $(N_v+1) \times
(N_v+1)$ real matrix. Precisely, the formation of $\bB$ is as
\begin{equation}\label{eq:eigenvectors_B}
\begin{split}
\bB&=\left[\begin{array}{ccc}
            \boldsymbol{\rm I}  & 0                   & \quad 0\quad \\
            \quad \bB_{21}\quad & \quad \bB_{22}\quad & \quad 0\quad \\
            \quad \bB_{31}\quad & \quad \bB_{32}\quad & \quad \boldsymbol{\rm I}\quad
        \end{array}\right]
    \begin{array}{l}
    \scriptstyle{\longleftarrow ~~~D \rm{~row}}\\
    \scriptstyle{\longleftarrow ~~~D(D-1)/2 \rm{~row}}\\
    \scriptstyle{\longleftarrow ~~~N_v-D(D+1)/2 \rm{~row}}
    \end{array}\\
&\quad~~~~~~~\uparrow \quad ~~~~~\uparrow \quad ~~~~ \uparrow\\
&\quad~~~~{\scriptstyle D \rm{~col}}\quad \scriptstyle{\frac{D(D-1)}{2} \rm{~col}}~~
    \scriptstyle{N_v-\frac{D(D+1)}{2} \rm{~col}}
\end{split}
\end{equation}
where $\boldsymbol{\rm I}$ is identity matrix, whose dimension is
context depended. The first $D$ rows of $\bB$ are arising from
\eqref{eq:eigenvectors_rho} and \eqref{eq:eigenvectors_u}, the
following $D(D-1)/2$ rows are arising from \eqref{eq:eigenvectors_p}
and \eqref{eq:eigenvectors_eij}, and the rest $N_v-D(D+1)/2$ rows are
arising from \eqref{eq:eigenvectors_all}.

The properties of $\bB$ are here further clarified. We denote the entry
of $\bB$ located at $(i,j)$ position as $b_{ij}$, whereas $i,j = 0$,
$\cdots$, $N_v$. Noticing that entries of $\bB_{21}$ are from
\eqref{eq:eigenvectors_p} and \eqref{eq:eigenvectors_eij}, we have
\begin{equation}\label{eq:riemann_b21}
\begin{aligned}
b_{ij}=\left\{\begin{array}{ll}
    \dfrac{\lambda^2}{2D}-\dfrac{1}{2},   &   \text{if }j=1, \text{ and
        } i=\mathcal{N}_{D-1}(2\hat{e}_l)\text{ for some }
        l=2,\cdots,D,\\ [2mm]
    0,      &   \text{else.}
    \end{array}\right.
\end{aligned}
\end{equation}
And \eqref{eq:eigenvectors_p} and \eqref{eq:eigenvectors_eij} make that
$\bB_{22}$ is a summation as 
\begin{equation}\label{eq:eigenvectors_b22}
 \bB_{22} = \boldsymbol{\rm I} + \frac{1}{D} {\bf\Omega}, 
\end{equation}
where ${\bf\Omega} \in \bbR^{\frac{D(D-1)}{2} \times
  \frac{D(D-1)}{2}}$ is a matrix with its entries $\omega_{ij}$ as
\begin{equation*}
\omega_{ij}=\left\{\begin{array}{ll}
                1,  &   \text{if }i=\mathcal{N}_{D-1}(2\hat{e}_k)-D,
                ~j=\mathcal{N}_{D-1}(2\hat{e}_l)-D \text{ for some
                }k,l\in\mathcal{D}\backslash\{1\},\\
                0,  &   \text{else}.
            \end{array}\right.
\end{equation*}
Since there are at most $D-1$ nonzero entries in each row of
${\bf\Omega}$, it is clear that ${\bf B}_{22}$ is strictly diagonally
dominant thus nonsingular, and then $\bB$ is nonsingular. With the
entry value of ${\bf \Omega}$ as given above, one can check ${\bf
  \Omega}^2=(D-1){\bf \Omega}$ holds. Hence we can get the inverse of
$\bB_{22}$, which reads:
\begin{equation}\label{eq:riemann_b22}
\bB_{22}^{-1}={\bf I}-\frac{1}{2D-1}{\bf\Omega}.
\end{equation}
Meanwhile, since $\bB$ is a nonsingular block lower triangular matrix, 
we can get its inverse as
\begin{equation}\label{eq:riemann_Binv}
\begin{split}
\bB^{-1}&=\left[\begin{array}{ccc}
            \boldsymbol{\rm I}     & 0             & \quad 0 \quad\\
            -\bB_{22}^{-1}\bB_{21} & \bB_{22}^{-1} & \quad 0 \quad\\
            *                      & *             & \quad \boldsymbol{\rm I}\quad
        \end{array}\right]
    \begin{array}{l}
    \scriptstyle{\longleftarrow ~~~D \rm{~row}}\\
    \scriptstyle{\longleftarrow ~~~D(D-1)/2 \rm{~row}}\\
    \scriptstyle{\longleftarrow ~~~N_v-D(D+1)/2 \rm{~row}}
    \end{array}\\
&\qquad~~~~~~\uparrow \quad ~~~~~~\uparrow \quad ~~~ \uparrow\\
&\qquad~~~~{\scriptstyle D \rm{~col}}\quad \scriptstyle{\frac{D(D-1)}{2} \rm{~col}}~~
    \scriptstyle{N_v-\frac{D(D+1)}{2} \rm{~col}}
\end{split}
\end{equation}
Let 
\begin{equation}
\hat{\bB}={\rm diag}\{\,{\bf I}, ~\bB_{22},~ {\bf I}\,\},
\end{equation}
be the diagonal blocks of $\bB$. The inverse of $\hat{\bB}$ is
$\hat{\bB}^{-1}={\rm diag}\{{\bf I},~\bB_{22}^{-1},~{\bf I}\}$, and we
have
\begin{equation}\label{eq:hbb_bbinv}
\hat{\bB} \bB^{-1}=\left[\begin{array}{ccc}
            \boldsymbol{\rm I} & 0       & \quad 0 \quad \\
            -\bB_{21}          & {\bf I} & \quad 0 \quad \\
            *                  & *       & \quad \boldsymbol{\rm I}\quad
        \end{array}\right].
\end{equation}
Since $\bB$ is nonsingular, for an arbitrary $j \in \{ 1, \cdots, N_v
\}$, we let
\begin{equation}\label{eq:eigenvector_value2}
\bv^{(j)}=\hat{\bB}\bB^{-1}I_j,
\end{equation}
where $I_j$ is the $j$-th column of the $N_v \times N_v$ identity
matrix. Actually, $\boldsymbol{v}^{(j)}$ is the $j$-th column of
$\hat{\bB}\bB^{-1}$. Notice in \eqref{eq:definition_zlambda} for any
$\alpha$ that $|\alpha|=M$, $\mathcal{N}_{D-1}(\hat{\alpha})$-th
component of $\boldsymbol{z}_\lambda$ is $\He_{\alpha_1+1}(\lambda)$.
For the $\alpha$ satisfying
\[
|\alpha|=M \text{ and } j=\mathcal{N}_{D-1}(\hat{\alpha}),
    ~k=\alpha_1+1,
\]
 we choose $\lambda$ such that 
\[
\He_{k}(\lambda) = 0.
\] 
Then we have that the $j$-th component of $\bz_\lambda$ vanishes
\[z_{\lambda,j} = \He_{k}(\lambda) = 0.\] This
makes \eqref{eq:zbv} valid that 
\begin{equation}\label{eq:eigenvector_value}
\bB\hat{\bB}^{-1}\boldsymbol{v}^{(j)}=\rho I_j,\quad z_{\lambda,j}=0, \quad
j=1,\cdots,N_v.
\end{equation}

Since $\br$ is depended only on $\bv$ and $\lambda$, we denote
$\br_{\hat{\alpha},i}$ to be the vector prescribed by the given
$\bv^{(j)}$, $j=\mathcal{N}_{D-1}(\hat{\alpha})$, and $\lambda =
\rC{i}{k}$, when $k=\alpha_1+1$, for arbitrary $|\alpha| = M$, $i=1$,
$\cdots$, $k$. It is clear that $\rC{i}{k}$ and $\br_{\hat{\alpha},i}$
are a pair of eigenvalue and eigenvector of $\btAM$ that
\[ \btAM \br_{\hat{\alpha},i} = \rC{i}{k} \br_{\hat{\alpha},i}. \]
The eigenvectors of $\btAM$ can be divided into a cluster of classes,
each of which is: for arbitrary $|\alpha|=M$,
\[
\{\br_{\hat{\alpha},i}\mid i=1,\cdots,k,~k=\alpha_1+1\}.
\]
This fact essentially stems from the reducibility of the matrix
$\btAM$.

Notice that
\begin{enumerate}
\item The components of $\bv$ are a subset of $\br$'s components,
  linearly independent $\bv^{(j)}$'s determine linearly independent
  $\br$'s;
\item Eigenvectors belongs to different eigenvalues are orthogonal and
  the $k$ zeros of Hermite polynomial $He_k(\lambda)$ are different.
\end{enumerate}
We have that $\br_{\hat{\alpha},i}$, $i=1,\cdots,k$, when $|\alpha|=M$
and $k=\alpha_1+1$ are linearly independent and the matrix $\btAM$ has
\begin{equation}
\sum_{k=1}^{M+1}k\binom{D-1+M-k}{M+1-k} = \binom{M+D}{D} = N
\end{equation}
linearly independent eigenvectors. 

On the other hand, for arbitrary $\br_{\hat{\alpha},i}$, there exists
a unique $\beta$ satisfying $\beta = (i-1)e_1+\tal$, hence,
there is a one-one mapping between $\br_{\hat{\alpha},i}$ and
$\alpha$ with $|\alpha|\le M$. So we can also get $\btAM$ has $N$ linearly
independent eigenvectors. This completes the proof.
\end{proof}

With the help of Lemma \ref{lem:numberofeigenvectors}, it is
not difficult to get the following result:
\begin{lemma}\label{lem:eigenpolynomial}
Let 
\begin{align}
\mathcal{\tilde{P}}_{1,m} &= \He_{m+1}(\lambda),\quad m\in\bbN,\\
\mathcal{\tilde{P}}_{D,m} &=
\prod_{k=0}^{m}\mathcal{\tilde{P}}_{D-1,k}, \quad 1<D\in\bbN^+.\label{eq:eigenpolynomials}
\end{align}
$\mathcal{\tilde{P}}_{D,M}$ is the characteristic polynomial of
 $\btAM$.
\end{lemma}
\begin{proof}
  In case of $D=1$, the result has been proved in \cite{Fan}. Here we give the
  proof for $D\ge 2$.  In the proof of lemma
  \ref{lem:numberofeigenvectors}, \eqref{eq:eigenvector_value} shows
  that the characteristic polynomial of $\btAM$ is
\begin{equation}\label{eq:eigenpolynomialD}
\mathcal{\tilde{P}}_{D,M}=\prod_{k=1}^{M+1}\He_k(\lambda)^{\binom{D-1+M-k}{M+1-k}}
   =\prod_{k=1}^{M+1}\He_k(\lambda)^{\binom{D-1+M-k}{D-2}}.
\end{equation}
Now we need only to prove that $\mathcal{\tilde{P}}_{D,M}$ satisfies
\eqref{eq:eigenpolynomials}. Here we use induction argument on $D$. As
$D=2$, \eqref{eq:eigenpolynomialD} can be written as
\begin{equation}
\mathcal{\tilde{P}}_{2,M}=
\prod_{k=1}^{M+1}\He_k(\lambda) =
\prod_{m=0}^{M}\mathcal{\tilde{P}}_{1,m}.
\end{equation}
We assume that \eqref{eq:eigenpolynomials} holds for $D-1, D\ge3$.
With the induction hypothesis, we have
\begin{equation}
\begin{aligned}
\prod_{m=0}^{M} \mathcal{\tilde{P}}_{D-1,m}
   &=\prod_{m=0}^{M}\left(\prod_{k=1}^{m+1}\He_k(\lambda)^{\binom{D-1-1+m-k}{D-1-2}}\right)\\
    &=\prod_{k=1}^{M+1}\He_k(\lambda)^{\sum_{m=k-1}^M{\binom{D-2+m-k}{D-3}}}\\
    &=\prod_{k=1}^{M+1}\He_k(\lambda)^{\binom{D-1+M-k}{D-2}}\\
    &=\mathcal{\tilde{P}}_{D,M}.
\end{aligned}
\end{equation}
This completes the proof.
\end{proof}
With the relation of $\btAM$ and $\bhAM$ \eqref{eq:dimensionless}, we
have the following theorem.
\begin{theorem}\label{thm:eigen}
Let 
\begin{align}\label{eq:def_eigenpolynomial}
\mathcal{{P}}_{1,m} &=
\He_{m+1}\left(\frac{\lambda-u_1}{\sqrt{\theta}}\right)\theta^{(m+1)/2},\quad m\in\bbN,\\
\mathcal{{P}}_{D,m} &=
\prod_{k=0}^{m}\mathcal{{P}}_{D-1,k}, \quad 1<D\in\bbN^+.
\end{align}
$\mathcal{{P}}_{D,M}$ is the characteristic polynomial of $\bhAM$.
And $\bhAM$ has N linearly independent eigenvectors, which read
\begin{equation}\label{eq:eigenvector_hA}
\hat{\br}_{\hat{\alpha},i} =
    \boldsymbol{\Lambda}^{-1}\br_{\hat{\alpha},i}, 
    \mbox{ for eigenvalue }\lambda_{i,k}=u_1+\rC{i}{k}\sqrt{\theta}
\end{equation}
for all $|\alpha|= M$, $i = 1, \cdots, k$, whereas $k=\alpha_1+1$.
\end{theorem}
\begin{proof}
Since 
\begin{equation}
\bhAM = u_1\boldsymbol{\mathrm{I}} +
    \sqrt{\theta}\boldsymbol{\Lambda}^{-1}\btAM\boldsymbol{\Lambda},
\end{equation}
and $\boldsymbol{\Lambda}$ is nonsingular, so any
$\br_{\hat{\alpha},i}\in\bbR^{N}$ is the eigenvector of $\btAM$ for the eigenvalue
$\rC{i}{k}$, then
$\hat{\br}_{\hat{\alpha},i}=\boldsymbol{\Lambda}^{-1}\br_{\hat{\alpha},i}$
is the eigenvector of $\bhAM$ for the eigenvalue
$u_1+\rC{i}{k}\sqrt{\theta}$. Using Lemma \ref{lem:eigenvectors}
and discussion in Lemma \ref{lem:numberofeigenvectors},
we obtain \eqref{eq:eigenvector_hA}. Lemma
\ref{lem:numberofeigenvectors} shows $\btAM$ has $N$ linearly
independent eigenvectors, so $\bhAM$ also has $N$ linearly
independent eigenvectors and \eqref{eq:eigenvector_hA} gives a set of
basis.

Lemma \ref{lem:numberofeigenvectors} and \eqref{eq:eigenvector_hA}
show that the characteristic polynomial of $\bhAM$ is
\begin{equation}
\mathcal{P}_{D,M}=\prod_{k=1}^{M+1}
    \left(\He_k\left(\frac{\lambda-u_1}{\sqrt{\theta}}\right)\theta^{k/2}\right)
    ^{\binom{D-1+M-k}{D-2}}.
\end{equation}
Similar as that in the proof of Lemma \ref{lem:eigenpolynomial}, 
$\mathcal{{P}}_{D,M}$ is thus the characteristic polynomial
of $\bhAM$.
\end{proof}
Theorem \ref{thm:diagonalisable} is now straightforward:
{\renewcommand\proofname{Proof of Theorem \ref{thm:diagonalisable}}
\begin{proof}
  With Theorem \ref{thm:eigen}, we declare that $\bhAM$ is
  diagonalisable with real eigenvalues directly, that is, the moment
  system \eqref{eq:hyperbolicsystem1d} is hyperbolic.
\end{proof}}


\section{System in Multi-dimensional Spatial Space} 
\label{sec:hyperbolicitymd} As the main result of this
paper, here we give the general hyperbolic moment system containing
all moments with orders lower than $M$. Without the assumption that
the dependence of $f$ on $x_2$, $\cdots$, $x_D$ is homogeneous,
according to the discussions in Section \ref{sec:momentsystem}, Grad's
moment system can be written in the following form:
\begin{equation}
\pd{\bw}{t} +
  \sum_{j = 1}^D {\bf M}_j(\bw) \pd{\bw}{x_j} = 0,
\end{equation}
where $\bw$ remains the same definition as the one-dimensional case
\eqref{eq:basic_moments}, and ${\bf M}_j$, $j = 1,\cdots,D$ are square
matrices depending on $\bw$. Comparing with \eqref{eq:1d}, one
immediately has ${\bf M}_1 = {\bf A}_M$. Similar as Definition
\ref{def:bhAM}, we give the following definition:
\begin{definition}
For $j = 1,\cdots,D$, $\hat{\bf M}_j$ is called the regularized matrix
of the matrix ${\bf M}_j$,  if it satisfies that for any admissible
$\bw$, 
\begin{equation} \label{eq:Grad_hme}
\hat{\bf M}_j \pd{\bw}{x_j} =
  {\bf M}_j \pd{\bw}{x_j} -
  \sum_{|\alpha|=M} \RM^j(\alpha) I_{\mN(\alpha)},
\end{equation}
where $I_k$ is the $k$-th column of the $N \times N$ identity matrix.
\end{definition}
Now the multi-dimensional regularized moment equations can be written
as
\begin{equation} \label{eq:hme}
  \pd{\bw}{t} + \sum_{j = 1}^D \hat{\bf M}_j(\bw)
  \pd{\bw}{ x_j} = 0.
\end{equation}
Recalling that
\begin{equation}\label{defRM2}
\RM(\alpha) = \sum_{j=1}^D \RM^j(\alpha),
\end{equation}
one finds that the multi-dimensional regularized moment system is
obtained by subtracting $\RM(\alpha)$ from \eqref{eq:Grad_hme}
for all $|\alpha| = M$. Applying such an operation on \eqref
{eq:momentseqs1}, we can reformulate the regularized moment system as
\begin{equation} \label{eq:moment_eqs}
\begin{split}
& \left( \pd{f_{\alpha}}{t} +
  \sum_{d=1}^D \pd{u_d}{t} f_{\alpha-e_{d}}
  + \frac{1}{2} \pd{\theta}{t}
    \sum_{d =1}^D f_{\alpha-2e_d}
\right) \\
& \qquad + \sum_{j=1}^D \left(
  \theta \pd{f_{\alpha-e_j}}{x_j} +
  u_j \pd{f_{\alpha}}{x_j} +
  (1 - \delta_{|\alpha|,M}) (\alpha_j+1)
    \pd{f_{\alpha+e_j}}{x_j}
\right) \\
& \qquad {} + \sum_{j=1}^D \sum_{d=1}^D
\pd{u_d}{x_j} \left(
  \theta f_{\alpha-e_d-e_j} + u_j f_{\alpha-e_d}
  + (1 - \delta_{|\alpha|,M}) (\alpha_j+1) f_{\alpha-e_d+e_j}
\right) \\
& \qquad {} + \frac{1}{2} \sum_{j=1}^D \sum_{d=1}^D
  \pd{\theta}{x_j}
  \left( \theta f_{\alpha-2e_d-e_j} +
  u_j f_{\alpha-2e_d} +
  (1 - \delta_{|\alpha|,M}) (\alpha_j+1) f_{\alpha-2e_d+e_j}
\right) = 0, \\
& \hspace{360pt} |\alpha| \leqslant M.
\end{split}
\end{equation}
Actually, \eqref{eq:moment_eqs} is away from \eqref{eq:hme} only by a
linear transformation due to \eqref{eq:conservation_laws} to eliminate
the time derivatives of $u_d$ and $\theta$. Precisely speaking, there
exists an invertible matrix ${\bf T}(\bw)$ depending on $\bw$ such
that \eqref{eq:moment_eqs} is identical to the following system:
\begin{equation} \label{eq:var_hme}
{\bf T}(\bw) \pd{\bw}{t} +
  \sum_{j = 1}^D {\bf T}(\bw) \hat{\bf M}_j(\bw)
    \pd{\bw}{x_j} = 0,
\end{equation}
If we let all partial derivatives with respect to $x_j$ with $j > 1$
to be zero, \eqref{eq:moment_eqs} reduces to the one-dimensional
hyperbolic moment system \eqref{eq:hyperbolicsystem1d} in Section
\ref{sec:hyperbolicms}. Comparison of \eqref{eq:hme} and
\eqref{eq:hyperbolicsystem1d} clearly shows that $\hat{\bf M}_1 =
\hat{\bf A}_M$.

The following theorem declares the hyperbolicity\footnote{For
  multi-dimensional quasi-linear systems, we refer the readers to
  \cite{FVM} for the definition of hyperbolicity.} of the
multi-dimensional regularized moment system \eqref{eq:hme}:
\begin{theorem} \label{thm:rot_inv} The regularized moment system
\eqref{eq:hme} is hyperbolic for any admissible $\bw$. Precisely, for
a given unit vector $\bn = (n_1, \cdots, n_D)$, there exists a constant
matrix $\bf R$ partially depending on $\bn$ that
\begin{equation} \label{eq:rot_inv}
\sum_{j=1}^D n_j \hat{\bf M}_j(\bw) =
  {\bf R}^{-1} \hat{\bf A}_M ({\bf R} \bw) {\bf R},
\end{equation}
and this matrix is diagonalizable with eigenvalues as
\begin{equation}
\bu \cdot \bn + \rC{n}{m} \sqrt{\theta},
  \qquad 1 \leqslant n \leqslant m \leqslant M+1.
\end{equation}
\end{theorem}

Actually, this theorem gives the rotation invariance of the
regularized moment system and its globally hyperbolicity. Since the
translation invariance of the system is apparent, it is concluded that
the regularized system is Galilean invariant. Precisely, if another
coordinates $(\tilde{x}_1, \dots, \tilde{x}_D)$ are chosen and the
vector $\bn$ is along the $\tilde{x}_1$-axis, then the rotated moment
system is equivalent to the original one. This result is easy to
understand: on one hand, Grad's moment system is rotationally
invariant, since the full $M$-degree polynomials are used in the
truncated Hermite expansion; on the other hand, our regularization is
symmetric in every direction, which can be considered as ``isotropic''
in some sense. However, a rigorous proof of this theorem is rather
tedious.

In the literature, two types of indices have been used in the moment
methods. In Grad's paper \cite{Grad}, indices such as
\begin{equation} \label{eq:ind_Grad}
\vartheta = (\vartheta_1, \cdots, \vartheta_m) \in \mathcal{D}^m
\end{equation}
is used to denote the $m$-th order moments, while in \cite{NRxx}, the
symbols
\begin{equation} \label{eq:ind_NRxx}
\alpha = (\alpha_1, \cdots, \alpha_D) \in \bbN^D
\end{equation}
is used as the subscripts of $|\alpha|$-th order moments. The former
is convenient for mathematical proofs, while the latter is easier to
use in the numerical implementation, since for \eqref{eq:ind_NRxx},
the map from the index set to the moment set is a bijection, while
this is not true for \eqref{eq:ind_Grad}. If \eqref{eq:ind_NRxx} and
\eqref{eq:ind_Grad} represent the same moment, then one has
\begin{equation} \label{eq:ind_rel}
\alpha = e_{\vartheta_1} + \cdots + e_{\vartheta_m}, \quad
  m = |\alpha|.
\end{equation}
Below, both types of indices are needed in the proof of rotation
invariance, and we will always use the variant forms of Greek letters
such as $\vartheta$ and $\varphi$ to denote the Grad-type indices, and
normal Greek letters such as $\alpha$ and $\beta$ will be used to
denote indices like \eqref{eq:ind_NRxx}. The Greek letter ``sigma''
denotes the conversion between them. Supposing \eqref{eq:ind_rel}
holds, we write
\begin{equation}
\alpha = \sigma(\vartheta), \quad \vartheta = \varsigma(\alpha).
\end{equation}
That is, the normal form of sigma $\sigma(\cdot)$ converts indices
like \eqref{eq:ind_Grad} to indices like \eqref{eq:ind_NRxx}, and the
variant form of sigma $\varsigma(\cdot)$ does the inverse conversion.
Note that for a given $\alpha$, the Grad-type index $\vartheta$
satisfying \eqref{eq:ind_rel} is not uniquely determined. Define
\begin{equation} \label{eq:bbSigma}
\bbSigma(\alpha) = \{
  \vartheta \in \mathcal{D}^{|\alpha|}
    \mid \sigma(\vartheta) = \alpha
\},
\end{equation}
and then in most cases, $\bbSigma(\alpha)$ has more than one element.
For example, if $D = 2$ and $\alpha = (2, 2)$, then
\begin{equation}
\bbSigma(\alpha) = \{(1, 1, 2, 2), (1, 2, 1, 2), (1, 2, 2, 1),
  (2, 2, 1, 1), (2, 1, 2, 1), (2, 1, 1, 2) \}.
\end{equation}
Thus $\varsigma(\alpha)$ has multiple values. However, there is always
one special element $\vartheta \in \bbSigma(\alpha)$ satisfying
\begin{equation}
\vartheta_1 \leqslant \cdots \leqslant \vartheta_{|\alpha|},
\end{equation}
and we use this element as the value of $\varsigma(\alpha)$. It is
easy to find
\begin{equation}
\sigma(\varsigma(\alpha)) = \alpha.
\end{equation}
Additionally, we use $\sigma_i(\vartheta)$ to denote the $i$-th
component of $\sigma(\vartheta)$.

Based on these symbols, we have the following lemma:
\begin{lemma} \label{lem:sum}
Suppose $\alpha \in \bbN^D$ and $F(\cdot)$ is a function on
$\mathcal{D}^m$. If $F$ satisfies that $F(\varphi)$ is zero when
$\sigma_i(\varphi) < \alpha_i$ for some $i \in \mathcal{D}$, then the
following equality holds:
\begin{equation}
\sum_{\varphi \in \mathcal{D}^m} F(\varphi)
  = \sum_{\substack{\beta \in \bbN^D\\|\beta| = m - |\alpha|}}
    \sum_{\varphi \in \bbSigma(\beta + \alpha)} F(\varphi).
\end{equation}
\end{lemma}
\begin{proof}
It is obvious that
\begin{equation}
\mathcal{I} \triangleq
  \bigcup_{\substack{\beta \in \bbN^D\\ |\beta| = m - |\alpha|}}
  \bbSigma(\beta + \alpha) \subset \mathcal{D}^m,
\end{equation}
and there are no duplicate elements in the union since $\bbSigma(\beta
+ \alpha) \cap \bbSigma(\tilde{\beta} + \alpha) = \emptyset$ if $\beta
\neq \tilde{\beta}$. Thus it only remains to prove that $\varphi \in
\mathcal{I}$ if
\begin{equation}
\varphi \in \mathcal{D}^m, \qquad \text{and}
  \qquad \sigma_i(\varphi) \geqslant \alpha_i,
  \quad \forall i \in \mathcal{D}.
\end{equation}
This is true since $\varphi \in \bbSigma(\beta + \alpha)$ for $\beta =
\sigma(\varphi) - \alpha$.
\end{proof}
As a special case of Lemma \ref{lem:sum}, we set $\alpha = 0$ and
have
\begin{equation} \label{eq:sum}
\sum_{\varphi \in \mathcal{D}^m} F(\varphi)
  = \sum_{\substack{\beta \in \bbN^D\\|\beta| = m}}
    \sum_{\varphi \in \bbSigma(\beta)} F(\varphi).
\end{equation}
Here $F(\cdot)$ is an arbitrary function on $\mathcal{D}^m$.

Some more symbols are introduced as follows. All $m$-permutations of
the set $\{1, \cdots, n\}$ form the following set:
\begin{equation} \label{eq:cal_A}
\mathcal{A}_n^m = \{
  \varpi = (\varpi_1, \cdots, \varpi_m) \in \{1, \cdots, n\}^m
  \mid \varpi_i \neq \varpi_j \text{ if } i \neq j
\}, \quad \forall m, n \in \bbN, \quad n \geqslant m,
\end{equation}
which contains $n! / m!$ elements. Thus when we want to construct a
short vector using the components of a long vector, we will use the
following notation:
\begin{equation}
\vartheta_{\varpi} =
  (\vartheta_{\varpi_1}, \cdots, \vartheta_{\varpi_m})
  \in \mathcal{D}^m, \quad \forall \vartheta \in \mathcal{D}^n,
  \quad \varpi \in \mathcal{A}_n^m.
\end{equation}
The remaining part is denoted as $\vartheta \backslash
\vartheta_{\varpi}$. For example, if $\vartheta = (1, 3, 2, 3, 1, 2,
1)$ and $\varpi = (5, 2, 4)$, then
\begin{equation}
\vartheta_{\varpi} = (1, 3, 3), \quad
  \vartheta \backslash \vartheta_{\varpi} = (1, 2, 2, 1).
\end{equation}

Below, ${\bf G} = (g_{ij})_{D\times D}$ stands for the rotation
matrix, and we suppose ${\bf G}$ is orthogonal and its the determinant
is $1$. Define
\begin{equation} \label{eq:Pi_g}
\Pi_g(\vartheta, \varphi) = \prod_{i=1}^n g_{\vartheta_i \varphi_i},
  \quad \forall \vartheta, \varphi \in \mathcal{D}^n,
\end{equation}
and then we have the following lemma:
\begin{lemma} \label{lem:ind}
For a given matrix $\bf G$ and multi-indices $\alpha, \beta \in
\bbN^D$, the following equality holds for arbitrary $\vartheta \in
\mathcal{D}^{|\alpha| + |\beta|}$:
\begin{equation} \label{eq:count}
\sum_{\varphi \in \bbSigma(\alpha + \beta)}
  \frac{\sigma(\varphi)!}{\sigma(\vartheta)!} \Pi_g(\vartheta,\varphi)
= \frac{\alpha!}{\sigma(\vartheta)!}
  \sum_{\varpi \in \mathcal{A}_{|\alpha| + |\beta|}^{|\beta|}}
    \Pi_g(\vartheta_{\varpi}, \varsigma(\beta))
  \sum_{\varphi \in \bbSigma(\alpha)}
    \Pi_g(\vartheta \backslash \vartheta_{\varpi}, \varphi).
\end{equation}
\end{lemma}
\begin{proof}
We first consider the case $|\beta| = 1$. Suppose $\beta = e_d$, and
then \eqref{eq:count} becomes
\begin{equation} \label{eq:start_case}
\sum_{\varphi \in \bbSigma(\alpha + e_d)}
  \frac{\sigma(\varphi)!}{\sigma(\vartheta)!} \Pi_g(\vartheta,\varphi) =
\frac{\alpha!}{\sigma(\vartheta)!}
  \sum_{i=1}^{|\alpha| + 1} g_{\vartheta_i d}
  \sum_{\varphi \in \bbSigma(\alpha)}
    \Pi_g(\vartheta \backslash \vartheta_i, \varphi).
\end{equation}
For $\varphi \in \bbSigma(\alpha + e_d)$, one has $\sigma(\varphi)! =
(\alpha_d + 1) \alpha!$. Thus \eqref{eq:start_case} is equivalent to
\begin{equation} \label{eq:prod_frac}
(\alpha_d + 1) \sum_{\varphi \in \bbSigma(\alpha + e_d)}
  \Pi_g(\vartheta, \varphi) =
\sum_{i=1}^{|\alpha| + 1} g_{\vartheta_i d}
  \sum_{\varphi \in \bbSigma(\alpha)}
    \Pi_g(\vartheta \backslash \vartheta_i, \varphi).
\end{equation}
For an arbitrary $\varphi \in \bbSigma(\alpha + e_d)$, if $\varphi_i =
d$, then $\Pi_g(\vartheta, \varphi) = g_{\vartheta_i d}
\Pi_g(\vartheta \backslash \vartheta_i, \varphi \backslash
\varphi_i)$, and $\varphi \backslash \varphi_i \in \bbSigma(\alpha)$.
Since there are $(\alpha_d + 1)$ choices of $i$ such that $\varphi_i =
d$, the product $\Pi_g(\vartheta, \varphi)$ appears $(\alpha_d + 1)$
times in the right hand side of \eqref{eq:prod_frac}. This proves
\eqref{eq:start_case}.

Suppose the lemma holds for $|\beta| = m - 1$, and we are going to
prove the case $|\beta| = m$. In order to use the
technique of induction, we choose $d \in \{1,\cdots,D\}$ such that
$\beta_d > 0$, and let $\beta' = \beta - e_d$. Thus $|\beta'| = m -
1$. Applying \eqref{eq:start_case}, one has
\begin{equation}
\begin{split}
\sum_{\varphi \in \bbSigma(\alpha + \beta)}
  \frac{\sigma(\varphi)!}{\sigma(\vartheta)!} \Pi_g(\vartheta,\varphi)
& = \sum_{\varphi \in \bbSigma(\alpha + \beta' + e_d)}
  \frac{\sigma(\varphi)!}{\sigma(\vartheta)!}
  \Pi_g(\vartheta,\varphi) \\
& = \frac{(\alpha + \beta')!}{\sigma(\vartheta)!}
  \sum_{i=1}^{|\alpha| + |\beta|} g_{\vartheta_i d}
  \sum_{\varphi \in \bbSigma(\alpha + \beta')}
    \Pi_g(\vartheta \backslash \vartheta_i, \varphi) \\
& = \sum_{i=1}^{|\alpha| + |\beta|}
  \frac{\sigma(\vartheta \backslash \vartheta_i)!}{\sigma(\vartheta)!}
  g_{\vartheta_i d} \sum_{\varphi \in \bbSigma(\alpha + \beta')}
    \frac{\sigma(\varphi)!}{\sigma(\vartheta \backslash \vartheta_i)!}
    \Pi_g(\vartheta \backslash \vartheta_i, \varphi)
\end{split}
\end{equation}
Defining $\hat{\vartheta}^i = \vartheta \backslash \vartheta_i$, and
using the inductive assumption, one obtains
\begin{equation} \label{eq:induction}
\begin{split}
\sum_{\varphi \in \bbSigma(\alpha + \beta)}
  \frac{\sigma(\varphi)!}{\sigma(\vartheta)!} \Pi_g(\vartheta,\varphi)
= \frac{\alpha!}{\sigma(\vartheta)!}
  \sum_{i=1}^{|\alpha| + |\beta|} g_{\vartheta_i d}
  \sum_{\varpi \in \mathcal{A}_{|\alpha|+|\beta|-1}^{|\beta|-1}}
    \Pi_g(\hat{\vartheta}_{\varpi}^i, \sigma(\beta'))
  \sum_{\varphi \in \bbSigma(\alpha)}
    \Pi_g(\hat{\vartheta}^i \backslash \hat{\vartheta}_{\varpi}^i, \varphi)
\end{split}
\end{equation}
It is evident that the right hand sides of \eqref{eq:count} and
\eqref{eq:induction} are the same. Thus the lemma is proved.
\end{proof}

Now let us start the rotation. We first define
\begin{equation} \label{eq:tilde_x}
\tilde{x}_i = \sum_{j=1}^D g_{ij} x_j, \quad i = 1,\cdots,D,
\end{equation}
and denote by $\tilde{\rho}$, $\tilde{\bu}$, $\tilde{\theta}$ the
density, macroscopic velocity and temperature in the new coordinates
$\tilde{\bx} = (\tilde{x}_1, \cdots, \tilde{x}_D)$. If we define
$\tilde{\bxi} = {\bf G} \bxi$, then the orthogonality of ${\bf G}$
shows
\begin{equation}
\begin{gathered}
\tilde{\rho} = \int_{\bbR^D} f(\bxi) \dd \tilde{\bxi}
  = \int_{\bbR^D} f(\bxi) \dd \bxi = \rho, \\
\tilde{\rho} \tilde{\bu}
  = \int_{\bbR^D} \tilde{\bxi} f(\bxi) \dd \tilde{\bxi}
  = \int_{\bbR^D} {\bf G} \bxi f(\bxi) \dd \bxi
  = \rho {\bf G} \bu, \\
\tilde{\rho} \tilde{\theta}
  = \frac{1}{D} \int_{\bbR^D}
    |\tilde{\bxi} - \tilde{\bu}|^2 f(\bxi) \dd \tilde{\bxi}
  = \frac{1}{D} \int_{\bbR^D}
    |\bxi - \bu|^2 f(\bxi) \dd \bxi
  = \rho \theta,
\end{gathered}
\end{equation}
and it follows immediately that
\begin{equation} \label{eq:tilde_theta_u}
\tilde{\theta} = \theta, \qquad
  \tilde{u}_i = \sum_{j=1}^D g_{ij} u_j, \quad i = 1,\cdots,D.
\end{equation}
Now we consider the general moments $\tilde{f}_{\alpha}$ in the
coordinates $\tilde{\bx}$. Define $\bz = (\bxi - \bu) / \sqrt{\theta}$
and $\tilde{\bz} = (\tilde{\bxi} - \tilde{\bu}) /
\sqrt{\tilde{\theta}}$. Then $\tilde{\bz} = {\bf G} \bz$. The
orthogonality of Hermite polynomials gives
\begin{equation} \label{eq:f_alpha_tilde_f_alpha}
\begin{gathered}
f_{\alpha} = \frac{(2\pi)^D \theta^{|\alpha| + D}}{\alpha!}
  \int_{\bbR^D} f(\bu + \sqrt{\theta} \bz)
    \mathcal{H}_{\theta,\alpha}(\bz)
    \exp \left( -\frac{|\bz|^2}{2} \right) \dd \bz, \\
\tilde{f}_{\alpha} =
  \frac{(2\pi)^D \tilde{\theta}^{|\alpha| + D}}{\alpha!}
  \int_{\bbR^D} f(\bu + \sqrt{\theta} \bz)
    \mathcal{H}_{\theta,\alpha}(\tilde{\bz})
    \exp \left( -\frac{|\tilde{\bz}|^2}{2} \right) \dd \tilde{\bz},
\end{gathered}
\end{equation}
From the definition of Hermite polynomials \eqref{eq:He}, it is easy
to find that \eqref{eq:cal_H} can be rewritten as
\begin{equation}
\mathcal{H}_{\theta,\alpha}(\bz) =
  (-1)^{|\alpha|} (2\pi)^{-\frac{D}{2}}
  \theta^{-\frac{{D}+|\alpha|}{2}}
  \frac{\partial^{|\alpha|}}{\partial^{\alpha} \bz}
  \exp \left( -\frac{|\bz|^2}{2} \right).
\end{equation}
Applying the chain rule of differentiation, we obtain
\begin{equation} \label{eq:H_tilde_v}
\begin{split}
\mathcal{H}_{\theta,\alpha}(\tilde{\bz}) & =
  (-1)^{|\alpha|} (2\pi)^{-\frac{D}{2}}
  \theta^{-\frac{{D}+|\alpha|}{2}}
  \sum_{\varphi \in \mathcal{D}^{|\alpha|}}
    \Pi_g(\varsigma(\alpha), \varphi)
  \frac{\partial^{|\alpha|}}{\partial^{\sigma(\varphi)} \bz}
  \exp \left( -\frac{|\bz|^2}{2} \right) \\
& = \sum_{\varphi \in \mathcal{D}^{|\alpha|}}
  \Pi_g(\varsigma(\alpha), \varphi)
  \mathcal{H}_{\theta, \sigma(\varphi)} (\bz).
\end{split}
\end{equation}
Collecting \eqref{eq:f_alpha_tilde_f_alpha} and \eqref{eq:H_tilde_v},
one has
\begin{equation} \label{eq:tilde_f_alpha}
\begin{split}
\tilde{f}_{\alpha} &= \frac{(2\pi)^D \theta^{|\alpha| + D}}{\alpha!}
  \sum_{\varphi \in \mathcal{D}^{|\alpha|}}
    \Pi_g(\varsigma(\alpha), \varphi)
    \int_{\bbR^D} f(\bu + \sqrt{\theta} \bz)
      \mathcal{H}_{\theta, \sigma(\varphi)} (\bz)
      \exp \left( -\frac{|\bz|^2}{2} \right)
    \dd \bz \\
&= \sum_{\varphi \in \mathcal{D}^{|\alpha|}}
  \frac{\sigma(\varphi)!}{\alpha!}
  \Pi_g(\varsigma(\alpha), \varphi) f_{\sigma(\varphi)}.
\end{split}
\end{equation}
As in \eqref{eq:basic_moments}, all the rotated moments can also be
collected into a vector denoted as $\tilde{\bw}$. The equations
\eqref{eq:tilde_f_alpha} and \eqref{eq:tilde_theta_u} directly give
the following result:
\begin{lemma}
Based on the expressions of the rotated moments
\eqref{eq:tilde_f_alpha} and \eqref{eq:tilde_theta_u}, the following
equalities hold for arbitrary $\alpha \in \bbN^D$:
\begin{gather}
\label{eq:u_dfdx}
\sum_{d=1}^D \tilde{u}_d
  \frac{\partial \tilde{f}_{\alpha}}{\partial \tilde{x}_d} =
\sum_{d=1}^D \sum_{\varphi \in \mathcal{D}^{|\alpha|}}
  \frac{\sigma(\varphi)!}{\alpha!} \Pi_g(\varsigma(\alpha), \varphi)
  \cdot u_d \frac{\partial f_{\sigma(\varphi)}}{\partial x_d}, \\
\label{eq:alpha_dfdx}
\sum_{d=1}^D (\alpha_d + 1)
  \frac{\partial \tilde{f}_{\alpha + e_d}}{\partial \tilde{x}_d} =
\sum_{d=1}^D \sum_{\varphi \in \mathcal{D}^{|\alpha|}}
  \frac{\sigma(\varphi)!}{\alpha!} \Pi_g(\varsigma(\alpha), \varphi)
  \cdot (\sigma_d(\varphi) + 1)
  \frac{\partial f_{\sigma(\varphi) + e_d}}{\partial x_d},
\end{gather}
where $\sigma_d(\varphi)$ is the $d$-th component of
$\sigma(\varphi)$.
\end{lemma}
\begin{proof}
Using \eqref{eq:tilde_f_alpha} and \eqref{eq:tilde_theta_u} directly,
we get
\begin{equation} \label{eq:u_dfdx_mid}
\sum_{d=1}^D \tilde{u}_d
  \frac{\partial \tilde{f}_{\alpha}}{\partial \tilde{x}_d} =
\sum_{d=1}^D \sum_{\varphi \in \mathcal{D}^{|\alpha|}}
  \frac{\sigma(\varphi)!}{\alpha!} \Pi_g(\varsigma(\alpha), \varphi)
  \cdot \sum_{j=1}^D g_{dj} u_j
  \frac{\partial f_{\sigma(\varphi)}}{\partial \tilde{x}_d}.
\end{equation}
Equation \eqref{eq:tilde_x} shows that
\begin{equation} \label{eq:partial_x}
\frac{\partial}{\partial x_j} =
  \sum_{d=1}^D g_{dj} \frac{\partial}{\partial \tilde{x}_d}.
\end{equation}
Thus \eqref{eq:u_dfdx} is the direct result of \eqref{eq:u_dfdx_mid}
and \eqref{eq:partial_x}.

The proof of \eqref{eq:alpha_dfdx} is also straightforward:
\begin{equation}
\begin{split}
\sum_{d=1}^D (\alpha_d + 1)
  \frac{\partial \tilde{f}_{\alpha + e_d}}{\partial \tilde{x}_d}
&= \sum_{d=1}^D (\alpha_d + 1)
  \sum_{\varphi \in \mathcal{D}^{|\alpha|+1}}
  \frac{\sigma(\varphi)!}{(\alpha + e_d)!}
  \Pi_g (\varsigma(\alpha + e_d), \varphi)
  \frac{\partial f_{\sigma(\varphi)}}{\partial \tilde{x}_d}. \\
&= \sum_{d=1}^D \sum_{j=1}^D g_{dj}
  \sum_{\varphi \in \mathcal{D}^{|\alpha|}}
    \frac{(\sigma(\varphi) + e_j)!}{\alpha!}
    \Pi_g(\varsigma(\alpha), \varphi)
  \frac{\partial f_{\sigma(\varphi) + e_j}}{\partial \tilde{x}_d} \\
&= \sum_{j=1}^D \sum_{\varphi \in \mathcal{D}^{|\alpha|}}
  \frac{\sigma(\varphi)!}{\alpha!} \Pi_g(\varsigma(\alpha), \varphi)
  \cdot (\sigma_j(\varphi) + 1)
  \frac{\partial f_{\sigma(\varphi) + e_j}}{\partial x_j}.
\end{split}
\end{equation}
This equality is identical to \eqref{eq:alpha_dfdx}.
\end{proof}

Using Lemma \ref{lem:ind}, it is not difficult to prove the following
lemma:
\begin{lemma} \label{lem:eqs1}
The following equalities hold for arbitrary $\alpha \in \bbN^D$:
\begin{equation}
\label{eq:dudt_f}
\sum_{d=1}^D \frac{\partial \tilde{u}_d}{\partial t}
  \tilde{f}_{\alpha - e_d} =
\sum_{d=1}^D \sum_{\varphi \in \mathcal{D}^{|\alpha|}}
  \frac{\sigma(\varphi)!}{\alpha!} \Pi_g(\varsigma(\alpha), \varphi)
  \frac{\partial u_d}{\partial t} f_{\sigma(\varphi) - e_d}, 
\end{equation}
\begin{equation}
\label{eq:dthetadt_f}
\sum_{d=1}^D \frac{\partial \tilde{\theta}}{\partial t}
  \tilde{f}_{\alpha - 2e_d} =
\sum_{d=1}^D \sum_{\varphi \in \mathcal{D}^{|\alpha|}}
  \frac{\sigma(\varphi)!}{\alpha!} \Pi_g(\varsigma(\alpha), \varphi)
  \frac{\partial \theta}{\partial t} f_{\sigma(\varphi) - 2e_d},
\end{equation}
\begin{equation}
\label{eq:dfdx}
\sum_{d=1}^D
  \frac{\partial \tilde{f}_{\alpha - e_d}}{\partial \tilde{x}_d} =
\sum_{d=1}^D \sum_{\varphi \in \mathcal{D}^{|\alpha|}}
  \frac{\sigma(\varphi)!}{\alpha!} \Pi_g(\varsigma(\alpha), \varphi)
  \frac{\partial f_{\sigma(\varphi) - e_d}}{\partial x_d},
\end{equation}
\begin{equation}
\label{eq:dudx_f}
\sum_{j=1}^D \sum_{d=1}^D
  \frac{\partial \tilde{u}_d}{\partial \tilde{x}_j}
  \tilde{f}_{\alpha - e_d - e_j} =
\sum_{j=1}^D \sum_{d=1}^D \sum_{\varphi \in \mathcal{D}^{|\alpha|}}
  \frac{\sigma(\varphi)!}{\alpha!} \Pi_g(\varsigma(\alpha), \varphi)
  \frac{\partial u_d}{\partial x_j}
  f_{\sigma(\varphi) - e_d - e_j},
\end{equation}
\begin{equation}
\label{eq:dthetadx_f}
\sum_{j=1}^D \sum_{d=1}^D
  \frac{\partial \tilde{\theta}}{\partial \tilde{x}_j}
  \tilde{f}_{\alpha - 2e_d - e_j} =
\sum_{j=1}^D \sum_{d=1}^D \sum_{\varphi \in \mathcal{D}^{|\alpha|}}
  \frac{\sigma(\varphi)!}{\alpha!} \Pi_g(\varsigma(\alpha), \varphi)
  \frac{\partial \theta}{\partial x_j} f_{\sigma(\varphi) - 2e_d - e_j}.
\end{equation}
\end{lemma}
\begin{proof}
Recalling that $f_{\beta} = 0$ if $\beta$ has a negative component and
using Lemma \ref{lem:sum}, we have the following equality:
\begin{equation}
\begin{split}
I & \triangleq\sum_{d=1}^D \sum_{\varphi \in \mathcal{D}^{|\alpha|}}
  \frac{\sigma(\varphi)!}{\alpha!} \Pi_g(\varsigma(\alpha), \varphi)
  \frac{\partial u_d}{\partial t} f_{\sigma(\varphi) - e_d} \\
& = \sum_{d=1}^D
  \sum_{\substack{\beta \in \bbN^D\\|\beta| = |\alpha| - 1}}
  \sum_{\varphi \in \bbSigma(\beta + e_d)}
  \frac{\sigma(\varphi)!}{\alpha!} \Pi_g(\varsigma(\alpha), \varphi)
  \frac{\partial u_d}{\partial t} f_{\beta}.
\end{split}
\end{equation}
Let $\vartheta = \varsigma(\alpha)$ and use \eqref{eq:start_case}, and
we have
\begin{equation}
I = \sum_{d=1}^D
  \sum_{\substack{\beta \in \bbN^D\\|\beta| = |\alpha| - 1}}
  \frac{\beta!}{\alpha!} \sum_{i=1}^{|\alpha|} g_{\vartheta_i d}
  \sum_{\varphi \in \bbSigma(\beta)}
    \Pi_g(\vartheta \backslash \vartheta_i, \varphi)
    \frac{\partial u_d}{\partial t} f_{\beta}.
\end{equation}
Now we employ \eqref{eq:sum} to join two of the summation symbols in
the above equation:
\begin{equation}
I = \sum_{d=1}^D \frac{\partial u_d}{\partial t}
  \sum_{i=1}^{|\alpha|} g_{\vartheta_i d}
  \sum_{\varphi \in \mathcal{D}^{|\alpha|-1}}
    \frac{\sigma(\varphi)!}{\alpha!}
    \Pi_g(\vartheta \backslash \vartheta_i, \varphi)
    f_{\sigma(\varphi)}.
\end{equation}
Using \eqref{eq:tilde_f_alpha} and \eqref{eq:tilde_theta_u} again, we
get
\begin{equation}
I = \sum_{d=1}^D \sum_{i=1}^{|\alpha|}
  \frac{1}{\alpha_{\vartheta_i}} g_{\vartheta_i d}
  \frac{\partial u_d}{\partial t} \tilde{f}_{\alpha - e_{\vartheta_i}}
= \sum_{i=1}^{|\alpha|} \frac{1}{\alpha_{\vartheta_i}}
  \frac{\partial \tilde{u}_{\vartheta_i}}{\partial t}
  \tilde{f}_{\alpha - e_{\vartheta_i}}
= \sum_{d=1}^D \frac{\partial \tilde{u}_d}{\partial t}
  \tilde{f}_{\alpha - e_d}.
\end{equation}
Thus \eqref{eq:dudt_f} is proved. The equation \eqref{eq:dthetadt_f}
can be proved in a similar way. Setting $\vartheta =
\varsigma(\alpha)$ and using Lemma \ref{lem:sum}, Lemma \ref{lem:ind}
and \eqref{eq:tilde_f_alpha}, we obtain
\begin{equation} \label{eq:II}
\begin{split}
\mathit{II} & \triangleq
  \sum_{d=1}^D \sum_{\varphi \in \mathcal{D}^{|\alpha|}}
  \frac{\sigma(\varphi)!}{\alpha!} \Pi_g(\varsigma(\alpha), \varphi)
  \frac{\partial \theta}{\partial t} f_{\sigma(\varphi) - 2e_d} \\
&= \sum_{d=1}^D
  \sum_{\substack{\beta \in \bbN^D\\|\beta| = |\alpha| - 2}}
  \sum_{\varphi \in \bbSigma(\beta + 2e_d)}
  \frac{\sigma(\varphi)!}{\sigma(\vartheta)!} \Pi_g(\vartheta,\varphi)
  \frac{\partial \theta}{\partial t} f_{\beta} \\
&= \sum_{d=1}^D
  \sum_{\substack{\beta \in \bbN^D\\|\beta| = |\alpha| - 2}}
    \frac{\beta!}{\alpha!}
  \sum_{\substack{i,j = 1,\cdots,|\alpha| \\ i\neq j}}
    g_{\vartheta_i d} g_{\vartheta_j d}
  \sum_{\varphi \in \bbSigma(\beta)}
    \Pi_g(\vartheta \backslash \vartheta_{(i,j)}, \varphi)
    \frac{\partial \theta}{\partial t} f_{\beta} \\
&= \sum_{d=1}^D \sum_{\substack{i,j = 1,\cdots,|\alpha| \\ i\neq j}}
  \frac{(\alpha - e_{\vartheta_i} - e_{\vartheta_j})!}{\alpha!}
  g_{\vartheta_i d} g_{\vartheta_j d}
  \frac{\partial \tilde{\theta}}{\partial t}
  \tilde{f}_{\alpha - e_{\vartheta_i} - e_{\vartheta_j}}.
\end{split}
\end{equation}
Since $\bf G$ is an orthogonal matrix, one has
\begin{equation}
\sum_{d=1}^D g_{\vartheta_i d} g_{\vartheta_j d} =
  \delta_{\vartheta_i \vartheta_j}.
\end{equation}
Thus \eqref{eq:II} can be further simplified as
\begin{equation}
\mathit{II} =
  \sum_{\substack{i,j = 1,\cdots,|\alpha| \\
    i\neq j, \, \vartheta_i = \vartheta_j}}
    \frac{1}{\alpha_{\vartheta_i} (\alpha_{\vartheta_i} - 1)}
    \frac{\partial \tilde{\theta}}{\partial t}
    \tilde{f}_{\alpha - 2e_{\vartheta_i}}
= \sum_{d=1}^D \frac{\partial \tilde{\theta}}{\partial t}
  \tilde{f}_{\alpha - 2e_d},
\end{equation}
which completes the proof of \eqref{eq:dthetadt_f}. The equations
\eqref{eq:dfdx}\eqref{eq:dudx_f}\eqref{eq:dthetadx_f} can be proved
using exactly the same technique. The detailed proofs are omitted here
to avoid redundancy.
\end{proof}

It is not difficult to find that \eqref{eq:dudt_f} and
\eqref{eq:dthetadt_f} still hold if we replace $t$ with $x_j$ or
$\tilde{x}_j$ for any $j = 1, \cdots, D$. Such observation leads to
the following two lemmas:
\begin{lemma}
The following equalities hold for arbitrary $\alpha \in \bbN^D$:
\begin{gather}
\label{eq:dudx_u_f}
\sum_{j=1}^D \sum_{d=1}^D
  \frac{\partial \tilde{u}_d}{\partial \tilde{x}_j}
  \tilde{u}_j \tilde{f}_{\alpha - e_d} =
\sum_{j=1}^D \sum_{d=1}^D \sum_{\varphi \in \mathcal{D}^{|\alpha|}}
  \frac{\sigma(\varphi)!}{\alpha!} \Pi_g(\varsigma(\alpha), \varphi)
  \frac{\partial u_d}{\partial x_j} u_j f_{\sigma(\varphi) - e_d}, \\
\label{eq:dthetadx_u_f}
\sum_{j=1}^D \sum_{d=1}^D
  \frac{\partial \tilde{\theta}}{\partial \tilde{x}_j}
  \tilde{u}_j \tilde{f}_{\alpha - 2e_d} =
\sum_{j=1}^D \sum_{d=1}^D \sum_{\varphi \in \mathcal{D}^{|\alpha|}}
  \frac{\sigma(\varphi)!}{\alpha!} \Pi_g(\varsigma(\alpha), \varphi)
  \frac{\partial \theta}{\partial x_j} u_j f_{\sigma(\varphi) - 2e_d}.
\end{gather}
\end{lemma}
\begin{proof}
Replacing $t$ with $x_j$ in \eqref{eq:dudt_f}, we obtain
\begin{equation}
\begin{split}
& \sum_{j=1}^D \sum_{d=1}^D \sum_{\varphi \in \mathcal{D}^{|\alpha|}}
  \frac{\sigma(\varphi)!}{\alpha!} \Pi_g(\varsigma(\alpha), \varphi)
  \frac{\partial u_d}{\partial x_j} u_j f_{\sigma(\varphi) - e_d}
= \sum_{j=1}^D \sum_{d=1}^D
  \frac{\partial \tilde{u}_d}{\partial x_j}
  u_j \tilde{f}_{\alpha - e_d} \\
={} & \sum_{j=1}^D \sum_{d=1}^D \sum_{i=1}^D 
  g_{ij} \frac{\partial \tilde{u}_d}{\partial \tilde{x}_i}
  u_j \tilde{f}_{\alpha - e_d}
= \sum_{d=1}^D \sum_{i=1}^D 
  \frac{\partial \tilde{u}_d}{\partial \tilde{x}_i}
  \tilde{u}_i \tilde{f}_{\alpha - e_d}.
\end{split}
\end{equation}
This equation is the same as \eqref{eq:dudx_u_f}. The proof of
\eqref{eq:dthetadx_u_f} is almost the same.
\end{proof}

\begin{lemma} \label{lem:eqs4}
The following equalities hold for arbitrary $\alpha \in \bbN^D$:
\begin{gather*}
\sum_{j=1}^D \sum_{d=1}^D (\alpha_j + 1)
  \frac{\partial \tilde{u}_d}{\partial \tilde{x}_j}
  \tilde{f}_{\alpha - e_d + e_j} =
\sum_{j=1}^D \sum_{d=1}^D \sum_{\varphi \in \mathcal{D}^{|\alpha|}}
  \frac{\sigma(\varphi)!}{\alpha!} \Pi_g(\varsigma(\alpha), \varphi)
  \cdot (\sigma_j(\varphi) + 1) \frac{\partial u_d}{\partial x_j}
  f_{\sigma(\varphi) - e_d + e_j}, \\
\sum_{j=1}^D \sum_{d=1}^D (\alpha_j + 1)
  \frac{\partial \tilde{\theta}}{\partial \tilde{x}_j}
  \tilde{f}_{\alpha - 2e_d + e_j} =
\sum_{j=1}^D \sum_{d=1}^D \sum_{\varphi \in \mathcal{D}^{|\alpha|}}
  \frac{\sigma(\varphi)!}{\alpha!} \Pi_g(\varsigma(\alpha), \varphi)
  \cdot (\sigma_j(\varphi) + 1) \frac{\partial \theta}{\partial x_j}
  f_{\sigma(\varphi) - 2e_d + e_j}.
\end{gather*}
\end{lemma}
\begin{proof}
Replacing $t$ with $\tilde{x}_j$ and substituting $\alpha + e_j$
for $\alpha$ in \eqref{eq:dudt_f}, one has
\begin{equation}
\begin{split}
& \sum_{j=1}^D \sum_{d=1}^D (\alpha_j + 1)
  \frac{\partial \tilde{u}_d}{\partial \tilde{x}_j}
  \tilde{f}_{\alpha - e_d + e_j} \\
={} & \sum_{j=1}^D \sum_{d=1}^D (\alpha_j + 1)
  \sum_{\varphi \in \mathcal{D}^{|\alpha| + 1}}
  \frac{\sigma(\varphi)!}{(\alpha + e_j)!}
  \Pi_g(\varsigma(\alpha + e_j), \varphi)
  \frac{\partial u_d}{\partial \tilde{x}_j}
  f_{\sigma(\varphi) - e_d} \\
={} & \sum_{j=1}^D \sum_{d=1}^D \sum_{i=1}^D g_{ji}
  \sum_{\varphi \in \mathcal{D}^{|\alpha|}}
  \frac{(\sigma(\varphi) + e_i)!}{\alpha!}
  \Pi_g(\varsigma(\alpha), \varphi)
  \frac{\partial u_d}{\partial \tilde{x}_j}
  f_{\sigma(\varphi) + e_i - e_d} \\
={} & \sum_{d=1}^D \sum_{i=1}^D
  \sum_{\varphi \in \mathcal{D}^{|\alpha|}}
  \frac{\sigma(\varphi)!}{\alpha!} \Pi_g(\varsigma(\alpha),\varphi)
  \cdot (\sigma_i(\varphi) + 1) \frac{\partial u_d}{\partial x_i}
  f_{\sigma(\varphi) + e_i - e_d}.
\end{split}
\end{equation}
This proves the first equality. The second equality can be similarly
proved, and the details are omitted.
\end{proof}

Now the proof of Theorem \ref{thm:rot_inv} is given as follows:

{\renewcommand\proofname{Proof of Theorem \ref{thm:rot_inv}}
\begin{proof}
  Since $(n_1, \cdots, n_D)$ is a unit vector, we let ${\bf G} =
  (g_{ij})_{D \times D}$ be an orthogonal matrix with its first row as
  $(n_1, \cdots, n_D)$. Now we use this matrix as the rotation matrix
  and define $\tilde{\bw}$ as \eqref{eq:tilde_f_alpha} and
  \eqref{eq:tilde_theta_u}. It is obvious that the relation between
  $\tilde{\bw}$ and $\bw$ is linear.  Therefore, there exists a
  constant matrix $\bf R$ (see \eqref{eq:tilde_f_alpha}) depending on
  ${\bf G}$ such that
\begin{equation}
\tilde{\bw} = {\bf R} \bw,
\end{equation}
and $\bf R$ is invertible since $\bw$ can be obtained from
$\tilde{\bw}$ by applying the rotation matrix ${\bf G}^{-1}$. Lemma
\ref{lem:eqs1}--\ref{lem:eqs4} have clearly shown that the ``rotated
equations''
\begin{equation} \label{eq:rotated_hme}
{\bf T}(\tilde{\bw}) \frac{\partial \tilde{\bw}}{\partial t} +
  \sum_{j = 1}^D {\bf T}(\tilde{\bw}) \hat{\bf M}_j(\tilde{\bw})
    \frac{\partial \tilde{\bw}}{\partial \tilde{x}_j} = 0
\end{equation}
can be deduced from \eqref{eq:var_hme} by linear operations. Thus
there exists a square matrix ${\bf H}(\bw)$ such that
\begin{equation} \label{eq:lt_hme}
{\bf H}(\bw) {\bf T}(\bw) \frac{\partial \bw}{\partial t} +
  \sum_{j = 1}^D {\bf H}(\bw) {\bf T}(\bw) \hat{\bf M}_j(\bw)
    \frac{\partial \bw}{\partial x_j} = 0
\end{equation}
is identical to \eqref{eq:rotated_hme}. Matching the terms with time
derivatives, one finds ${\bf H}(\bw) = {\bf T}(\tilde{\bw}) {\bf R}
{\bf T}(\bw)^{-1}$. Thus \eqref{eq:lt_hme} becomes
\begin{equation}
{\bf T}(\tilde{\bw}) \frac{\partial \tilde{\bw}}{\partial t}
  + \sum_{j=1}^D {\bf T}(\tilde{\bw}) {\bf R} \hat{\bf M}_j(\bw)
    \frac{\partial \bw}{\partial x_j} = 0.
\end{equation}
Using \eqref{eq:partial_x}, the above equation can be rewritten as
\begin{equation}
{\bf T}(\tilde{\bw}) {\bf R} \frac{\partial \bw}{\partial t}
  + \sum_{j=1}^D \sum_{d=1}^D
    g_{dj} {\bf T}(\tilde{\bw}) {\bf R} \hat{\bf M}_j(\bw)
    \frac{\partial \bw}{\partial \tilde{x}_d} = 0.
\end{equation}
Compared with \eqref{eq:rotated_hme}, one concludes
\begin{equation}
\sum_{j=1}^D g_{1j} {\bf T}(\tilde{\bw}) {\bf R} \hat{\bf M}_j(\bw)
  = {\bf T}(\tilde{\bw}) \hat{\bf M}_1({\bf R} \bw) {\bf R}.
\end{equation}
Multiplying both sides by ${\bf R}^{-1} {\bf T}(\tilde{\bw})^{-1}$,
\eqref{eq:rot_inv} is attained. Recalling $\hat{\bf M}_1 = \hat{\bf
A}_M$ and that the first component of the macroscopic velocity after
the rotation is $\bu \cdot \bn$ (see \eqref{eq:tilde_theta_u}), the
diagonalizability and the eigenvalues of the matrix \eqref{eq:rot_inv}
are naturally obtained using Theorem \ref{thm:eigen}.
\end{proof}}


\section{Riemann Problem} \label{sec:rp} 

Though the regularized moment system \eqref{eq:hme} is given by the
moment expansion up to an arbitrary order $M$ thus extremely complex, we
can clarify appreciably the structures of the elementary waves of this
system with Riemann initial value, including the rarefaction wave,
contact discontinuity and shock wave. Definitely, the structure of
the elementary wave is fundamental for further investigation into the
behavior of the solution of the system. Furthermore, the solution
structure of the Riemann problem is instructional for studying the
approximate Riemann solver, which is the basis of the numerical
methods using Godunov type schemes. The analysis below shows that
the structure of the elementary wave of the Riemann problem is quite
natural an extension of that of Euler equations, which indicates that
the regularized moment system \eqref{eq:hme} is actually a very
reasonable high order moment approximation of Boltzmann equation.
Following \cite{Toro} where the multi-dimensional Euler equations are
studied, we consider the $x_1$-split, $D$-dimensional Riemann
problem as below:
\begin{equation}\label{eq:splitRiemannProblem}
\left\{
\begin{aligned}
&\pd{\bw}{t}+\bhAM\pd{\bw}{x_1}=0,\\
&\bw(x_1,t=0)=\left\{\begin{array}{ll}
                \bw_L   &   \text{if } x_1<0,\\ [2mm]
                \bw_R   &   \text{if } x_1>0.
                \end{array}\right.
\end{aligned}
\right.
\end{equation}
The Riemann problem with 1D velocity space has been studied in
\cite{Fan} in detail. Here we focus on the case of $D \ge 2$.

Let us first recall the definition of the notations $\btAM$, $\bhAM$,
$\bB$, $\hat{\bB}$, $\bw$, $\br_{\hat{\alpha},i}$,
$\hat{\br}_{\hat{\alpha},i}$, $\bv^{(j)}$ and $\lambda_{i,k}$ in
Section \ref{sec:hyperbolicms}. In particular, we need the expressions
of $\hat{\bB}\bB^{-1}$ and $\bv^{(j)}$, which read
\begin{displaymath}\tag{\ref{eq:hbb_bbinv}'}\label{eq:hbb_bbinv2}
 \hat{\bB} \bB^{-1}=\left[\begin{array}{ccc}
            \boldsymbol{\rm I} & 0       &  \quad 0\quad \\
            -\bB_{21}          & {\bf I} &  \quad 0\quad \\
            *                  & *       &  \quad \boldsymbol{\rm I}\quad
        \end{array}\right]
\end{displaymath}
and for any $|\alpha|=M$,
    $k=\alpha_1+1,~j=\mathcal{N}_{D-1}(\hat{\alpha})$,
\begin{displaymath}\tag{\ref{eq:eigenvector_value2}'}\label{eq:eigenvector_value22}
\boldsymbol{v}^{(j)}=\hat{\bB}\bB^{-1}I_j,\quad \He_{k}({\rm C})=0,
\end{displaymath}
where $I_j$ is the $j$-th column of the $N_v \times N_v$ identity
matrix. $\hat{\br}_{\hat{\alpha},i}$, which depends on $\bv^{(j)}$ and
$\lambda_{i,k}=u_1+\rC{i}{k}\sqrt{\theta}$, is the eigenvector of
$\bhAM$ for the eigenvalue $\lambda_{i,k}=u_1+\rC{i}{k}\sqrt{\theta}$,
where $j=\mathcal{N}_{D-1}(\hat{\alpha}),~k=M+1-|\hat{\alpha}|$.  As
the first conclusion on the Riemann problem
\eqref{eq:splitRiemannProblem}, we have the following theorem:
\begin{theorem}\label{thm:elementarywave}
Each characteristic field of \eqref{eq:splitRiemannProblem} is either
genuinely nonlinear or linearly degenerate. And one characteristic
field is genuinely nonlinear if and only if $\bv$ (determined by the
right eigenvector through \eqref{eq:def_v}) and the eigenvalue
$\lambda=u_1+{\rm C}\sqrt{\theta}$ satisfy one of the following two
conditions:
\begin{enumerate}
\item $\bv=\bv^{(1)}$, and ${\rm C}$ subject to $\He_{M+1}({\rm C}) =
  0$ and ${\rm C} \neq 0$;
\item $\bv=\bv^{(j)}$, $j=\mathcal{N}_{D-1}(2\hat{e}_k)$,
  $k\in\mathcal{D}\backslash\{1\}$, and ${\rm C}$ subject to
  $\He_{M-1}({\rm C})=0$ and ${\rm C}\neq 0$.
\end{enumerate}
\end{theorem}
\begin{proof}
Let $\hat{\br}$ denote an eigenvector of $\bhAM$ with the eigenvalue
$\lambda=u_1+\rmC\sqrt{\theta}$ and $\bv$ is the corresponding vector
determined by \eqref{eq:def_v}. Since 
\begin{equation}
\lambda=u_1+\rmC\sqrt{\theta}=
  u_1+\rmC\sqrt{\frac{\sum_{d=1}^Dp_{2e_d}}{D\rho}}
\end{equation}
depends only on $\rho$, $u_1$ and $p_{2e_d}/2$, $d\in\mathcal{D}$, we
have
\begin{equation}\label{eq:condition}
\begin{split}
\nabla_{\bw}\lambda\cdot \hat{\br} &= 
    -\frac{\rmC\sqrt{\theta}}{2\rho}\cdot \rho r_{\rho}
    +1\cdot{\rmC\sqrt{\theta}}r_{\rho}
    +\frac{\rmC}{D\rho\sqrt{\theta}}\cdot\frac{\rmC^2\theta}{2}\rho r_{\rho}
    +\sum_{d=2}^D\frac{\rmC}{D\rho\sqrt{\theta}}\cdot\rho\theta r_{p_{2e_d}/2}\\
&=\frac{\sqrt{\theta}\rmC}{2}\left[\left(1+\frac{\rmC^2}{D}\right)v_1
    +\sum_{d=2}^D\frac{2}{D}v_{\mathcal{N}_{D-1}(2\hat{e}_d)}\right].
\end{split}
\end{equation}
\begin{itemize}
\item
If $\bv=\bv^{(1)}$, then \eqref{eq:eigenvector_value22} shows
$\bv^{(1)}=\hat{\bB}\bB^{-1}I_1$ and $\He_{M+1}({\rm C})=0$. From
\eqref{eq:hbb_bbinv2}, we get
\begin{equation}
v_1=1\text{ and }
v_{\mathcal{N}_{D-1}(2\hat{e}_2)}=\cdots=v_{\mathcal{N}_{D-1}(2\hat{e}_D)}
= \frac{\rmC^2}{2D}-\frac{1}{2}.
\end{equation}
Thus \eqref{eq:condition} can be written as
\begin{equation}
\nabla_{\bw}\lambda\cdot \hat{\br}=  
\frac{(D+1)\sqrt{\theta}\rmC}{4D^2}\left(\rmC^2+D\right).
\end{equation}
Hence,
\[
\left\{
\begin{array}{ll}
  \nabla_{\bw}\lambda\cdot \hat{\br} \equiv 0,     & \rm{if~} \rmC = 0, \\ [2mm]
  \nabla_{\bw}\lambda\cdot \hat{\br} \not\equiv 0, & \rm{otherwise.}
\end{array}\right.
\]
\item
If $\bv=\bv^{(j)},~j=\mathcal{N}_{D-1}(2\hat{e}_k)$ for any
$k\in\mathcal{D}\backslash\{1\}$, \eqref{eq:eigenvector_value22}
shows $\bv=\hat{\bB}\bB^{-1}I_j$ and $\He_{M-1}(\rmC)=0$. 
From \eqref{eq:hbb_bbinv2}, we can get
\[
v_j=1 \text{ and } 
v_1=v_{l}=0, ~l=\mathcal{N}_{D-1}(2\hat{e}_d)\text{ for any
    $d\in\mathcal{D}\backslash\{1,k\}$}.
\]
Then \eqref{eq:condition} is simplified as
\begin{equation}
\nabla_{\bw}\lambda\cdot \hat{\br}_{\hat{\alpha},i} =  
    \frac{\sqrt{\theta}}{D}\rmC.
\end{equation}
Again, we have
\[
\left\{
\begin{array}{ll}
  \nabla_{\bw}\lambda\cdot \hat{\br} \equiv 0,     & \rm{if~} \rmC = 0, \\ [2mm]
  \nabla_{\bw}\lambda\cdot \hat{\br} \not\equiv 0, & \rm{otherwise.}
\end{array}\right.
\]
\item
Otherwise, \eqref{eq:hbb_bbinv2} indicates
$v_1=v_{\mathcal{N}_{D-1}(2\hat{e}_k)}=0$ for each
$k\in\mathcal{D}\backslash\{1\}$.
Hence $\nabla_{\bw}\lambda\cdot \hat{\br}\equiv0$
always holds.
\end{itemize}
This completes the proof.
\end{proof}
This theorem reveals that for each characteristic field, the
eigenvalue is constant or varies monotonically along the integral
curve, resulting in simple wave structures. Below, some elementary
waves including the rarefaction waves, contact discontinuities and
shock waves are studied in detail, and the basic relations across
these waves are established.

The analysis below is based on the fact that an eigenvector
$\hat{\br}$ of $\bhAM$ for the eigenvalue $\lambda = u_1 +
\rmC\sqrt{\theta}$ depends only on $\bv$ and $\rmC$. With Theorem
\ref{thm:elementarywave} and the forms of $\bB$ and $\hat{\bB}\bB$
in Lemma \ref{lem:numberofeigenvectors}, we can divide characteristic
fields into three cases:
\begin{description}
\item[Case 1:]\label{enu:case1}
\begin{math}
\bv=\bv^{(1)},\text{ and ${\rm C}$ subject to }\He_{M+1}({\rm C}) = 0,
    \text{ and } {\rm C}\neq 0.
\end{math}
\item[Case 2:]\label{enu:case2}
\begin{math}
\bv=\bv^{(j)},~ j=\mathcal{N}_{D-1}(2\hat{e}_k)\text{ for any
}k\in\mathcal{D}\backslash\{1\}, \text{ and $\rmC$ subject to }
\He_{M-1}({\rm C})=0, \text{ and } {\rm C}\neq 0.
\end{math}
\item[Case 3:]\label{enu:case3} otherwise.
\end{description}

For convenience, let characteristic field $\alpha$ denote the
characteristic field corresponding to the eigenvector
$\hat{\br}_{\hat{\alpha},i}$ for the eigenvalue
$\lambda_{i,k}=u_1+\rC{i}{k}\sqrt{\theta}$ with
$i=\alpha_1$, $k=M+1-|\hat{\alpha}|$. Below, the rarefaction waves,
contact discontinuities and shock waves will be studied respectively.
\subsection{Rarefaction waves}\label{sec:rarefactionwave}
For the regularized moment system, if two states $\bw^L$ and
$\bw^R$ are connected by a rarefaction wave in a genuinely
nonlinear field $\alpha$, then the following two conditions must be
met:
\begin{itemize}
\item constancy of the \emph{generalised Riemann invariants} across
  the wave, saying the integral curve $\tilde{\bw}(\zeta) =
  (\tilde{w}_1(\zeta), \tilde{w}_2(\zeta), \cdots,
  \tilde{w}_N(\zeta))$ in the $N$-dimensional phase space satisfies
\begin{equation}\label{eq:RiemannInvariant}
\tilde{\bw}'(\zeta)=\hat{\br}_{\hat{\alpha}, i}(\tilde{\bw}),
\end{equation}
with $i=\alpha_1$.
\item
divergence of characteristics
\begin{equation}\label{eq:rarefaction_characteristics}
\lambda_{i,k}({\bw^L})<\lambda_{i,k}({\bw^R}).
\end{equation}
\end{itemize}
Fortunately, for a given point $\bw^0=(\rho^0, u_1^0, \cdots, w_j^0,
\cdots, w_N^0)$ in the phase space, the integral curve across $\bw^0$
can be given. Since $p=\displaystyle \frac{1}{D} \sum_{d=1}^D
p_{2e_d}$, we let $p^0=\displaystyle \frac{1}{D} \sum_{d=1}^D
p^0_{2e_d}$. The results are rather tedious, and here the integral
curves are only partially given in three cases as below:
\begin{itemize}
\item
If $\bv=\bv^{(1)}$, we have
\[
r_{\rho}=\rho,\quad r_{u_1}=\rC{i}{k}\sqrt{\theta},\quad
r_{u_d}=0,
\]\[
r_{p_{2e_1}/2}=\frac{\rC{i}{k}^2}{2}\rho\theta,\quad
r_{p_{2e_d}/2}=\frac{\rC{i}{k}^2-D}{2D}\rho\theta,
    \quad d\in\mathcal{D}\backslash\{1\}.
\]
Let $\Gamma=\dfrac{D-1+\rC{i}{k}^2}{2D-1}$, and then we have
\begin{subequations}\label{eq:rarefactionequation1}
\begin{align}
\tilde{\rho}(\zeta)&=\rho^0\exp(\zeta),\\
\tilde{u}_1(\zeta)&=u_1^0+\frac{2\rC{i}{k}\sqrt{\theta^0}}{\Gamma-1}
    \left[\exp\left(\frac{\Gamma-1}{2}\zeta\right)-1\right],\\
\tilde{u}_d(\zeta)&=u_d^0,\quad d=2,\cdots,D,\\
\tilde{p}_{2e_1}(\zeta)&=p_{2e_1}^0+\frac{\rC{i}{k}^2-D}{D\Gamma}p^0
    \left[\exp\left(\Gamma\zeta\right)-1\right],\\
\tilde{p}(\zeta)&=p^0\exp\left(\Gamma\zeta\right).\label{eq:rarefactionequation1_p}
\end{align}
\end{subequations}
\item
If $\bv=\bv^{(j)}$, $j=\mathcal{N}_{D-1}(2\hat{e}_k)$,
   $k\in\mathcal{D}\backslash\{1\}$, we have
\[
r_{\rho}=0,\quad, r_{u_d}=0,~d\in\mathcal{D},\quad
r_{p_{2e_k}/2} = \rho\theta,\quad
r_{p_{2e_d}/2}=0,~d\in\mathcal{D}\backslash\{k\}.
\]
Hence, the integral curve satisfies
\begin{subequations}\label{eq:rarefactionequation2}
\begin{align}
\tilde{\rho}(\zeta)&=\rho^0,\\
\tilde{u}_d(\zeta)&=u_d^0,\quad d=1,\cdots,D,\\
\tilde{p}_{2e_1}(\zeta)&=p_{2e_1}^0,\\
\tilde{p}(\zeta)&=p^0\exp\left(\frac{2\zeta}{D}\right).\label{eq:rarefactionequation2_p}
\end{align}
\end{subequations}
\item
Otherwise ($\bv=\bv^{(j)}$, $j\neq \mathcal{N}_{D-1}(2\hat{e}_k)$ for
any $k\in\mathcal{D}$),
\[
r_{\rho}=r_{u_1}=r_{p_{2e_d}/2}=0,~~d\in\mathcal{D}.
\]
Hence, we have
\begin{subequations}\label{eq:rarefactionequation3}
\begin{align}
\tilde{\rho}(\zeta)&=\rho^0,    &   \tilde{u}_1(\zeta)&=u_1^0,\\
\tilde{p}_{2e_1}(\zeta)&=p_{2e_1}^0,    &\tilde{p}(\zeta)&=p^0.
\end{align}
\end{subequations}
\end{itemize}
One can check that \eqref{eq:rarefactionequation1},
\eqref{eq:rarefactionequation3} and \eqref{eq:rarefactionequation2}
satisfy \eqref{eq:RiemannInvariant}. And an eigenvalue, which
satisfies \eqref{eq:eigenvector_hA}, of $\bhAM(\tilde{\bw}(\zeta))$ is
as
\begin{equation*}
\begin{split}
&\qquad s_{i,k}(\tilde{\bw}(\zeta))=\tilde{u}_1(\zeta)+
 \rC{i}{k}\sqrt{\tilde{p}(\zeta)/\tilde{\rho}(\zeta)} \\ 
&=\left\{\begin{array}{l}
    u_1^0+\rC{i}{k}\sqrt{\theta^0} +
\dfrac{\Gamma+1}{\Gamma-1}\rC{i}{k}\sqrt{\theta^0}\left[\exp\left(\dfrac{\Gamma-1}{2}\zeta\right)-1\right],
    \quad \text{for $\bv^{(1)}$},\\ [4mm]
    u_1^0+\rC{i}{k}\sqrt{\theta^0}\exp\left(\dfrac{\zeta}{D}\right),
    \quad ~{\text{for $\bv^{(j)}$, $j=\mathcal{N}_{D-1}(2\hat{e}_k)$,
   $k\in\mathcal{D}\backslash\{1\}$}},\\ [4mm]
    u_1^0+\rC{i}{k}\sqrt{\theta^0}, \qquad\qquad\qquad \text{otherwise}.
    \end{array}\right.
\end{split}
\end{equation*}
It is convenient to verify that $s_{i,k}(\tilde{\bw}(\zeta))\gtrless
s_{i,k}(\bw^0)$ if and only if $\rC{i}{k}\zeta\gtrless 0$, and $\bv$
and $\rC{i}{k}$ satisfy case 1 or case 2. Therefore, if the left state
$\bw^L$ and the right state $\bw^R$ are connected by a rarefaction
wave and let $\bw^0=\bw^L$, \eqref{eq:rarefaction_characteristics}
indicates $s_{i,k}(\bw^L)<s_{i,k}(\bw^R)$, hence $\rC{i}{k}\zeta>0$
and $\bv, \rC{i}{k}$ satisfies case 1 or case 2.
Therefore, we have that
\begin{itemize}
\item
for case 1:
$$u_d^L=u_d^R,\quad d=2,\cdots,D,$$
and 
\begin{subequations}
\begin{align}
\text{if } \rC{i}{k}>0, \text{ then } u_1^L<u_1^R, \quad
    p^L<p^R,\\
\text{if } \rC{i}{k}<0, \text{ then } u_1^L<u_1^R, \quad
    p^L>p^R.
\end{align}
\end{subequations}
\item
for case 2:
$$u_d^L=u_d^R,\quad d=2,\cdots,D,$$
and 
\begin{subequations}
\begin{align}
\text{if } \rC{i}{k}>0, \text{ then } u_1^L=u_1^R, \quad
    p^L<p^R,\\
\text{if } \rC{i}{k}<0, \text{ then } u_1^L=u_1^R, \quad
    p^L>p^R.
\end{align}
\end{subequations}
\end{itemize}
\subsection{Contact discontinuities}
For a contact discontinuities, \eqref{eq:RiemannInvariant} is still valid, and
the divergence of characteristics is replaced by
\begin{equation}\label{eq:contact_characteristics}
\lambda_{i,k}(\bw^L)=\lambda_{i,k}(\bw^R).
\end{equation}
According to Theorem \ref{thm:elementarywave} and analysis in
Section \ref{sec:rarefactionwave}, the contact discontinuities can be
founded if and only if $\bv$ and $\rC{i}{k}$ satisfy case 3.
\begin{itemize}
\item For $\bv^{(1)}$, \eqref{eq:contact_characteristics} means
  $\rC{i}{k}=0$. Substituting it into \eqref{eq:rarefactionequation1},
  we can get $u_d$, $d\in\mathcal{D}$ are invariant, while $p$,
  $p_{2e_1}$ are not (otherwise, \eqref{eq:rarefactionequation1_p}
  gives us $\zeta=0$, thus $\bw^L=\bw^R$).
\item For $\bv^{(j)}$, $j=\mathcal{N}_{D-1}(2\hat{e}_k)$,
  $k\in\mathcal{D}\backslash\{1\}$, \eqref{eq:contact_characteristics}
  means $\rC{i}{k} =0$ again. \eqref{eq:rarefactionequation2} shows
  $\rho$, $u_d$, $d\in\mathcal{D}$, $p_{2e_1}$ are invariant, while
  $p$ is not (otherwise, \eqref{eq:rarefactionequation2_p} gives us
  $\zeta=0$, which results $\bw^L=\bw^R$).
\item Otherwise, \eqref{eq:rarefactionequation3} shows $u_1$, $p$,
  $p_{2e_1}$ are all invariant.
\end{itemize}
Summarizing the discussion above, we conclude that if
$\rC{i}{k}\neq0$, then $u_1$, $p$, $p_{2e_1}$ are invariant across the
contact discontinuities, while if $\rC{i}{k}=0$, $u_1$ is invariant
and $p$ is not. However, $u_d$, $d=2, \cdots, D$ may change
discontinuously across a contact discontinuity. In fact, the case $\bv
= \bB^{-1}I_d$, $d=2, \cdots, D$ corresponds to a contact discontinuity
where $u_d$ is discontinuous. This is similar as the Euler equations.
\subsection{Shock waves}
The discussion of the shock wave needs some more scrupulosity. As is
well known, the jump condition on the shock wave is sensitive to the
form of the hyperbolic equations. Thus, before we give the
Rankine-Hugoniot condition, it is necessary to rewrite
\eqref{eq:splitRiemannProblem} in an appropriate form. However,
\eqref{eq:splitRiemannProblem} cannot be written as conservation
laws due to the presence of $\RM^1(\alpha)$. Nevertheless,
\eqref{eq:splitRiemannProblem} can still keep the conservation of the
conservative moments with orders from $0$ to $M-1$. Therefore,
\eqref{eq:splitRiemannProblem} can be reformulated by $\mN((M-1)e_D)$
conservation laws and $N-\mN((M-1)e_D)$ non-conservative equations.

Let
\begin{equation}
\bF=(F_{0},
            F_{e_1},F_{e_2},\cdots,F_{Me_D}),\quad 
F_{\alpha}=\frac{1}{\alpha!}\int_{\bbR^D}\bxi^\alpha f\dd\bxi,
    \quad |\alpha|\le M,
\end{equation}
where $\bxi^\alpha=\prod_{d=1}^D\xi_d^{\alpha_d}$ and $F_{0}$ stands
for $F_{\alpha}|_{\alpha=\bf{0}}$. Then \eqref{eq:splitRiemannProblem}
can be written as
\begin{equation}\label{eq:conservation_F}
\begin{aligned}
\pd{F_{\alpha}}{t}&+(\alpha_1+1)\pd{F_{\alpha+e_1}}{x_1}=0,\quad
    |\alpha|<M,\\
\pd{F_{\alpha}}{t}&+(\alpha_1+1)\pd{\hat{F}_{\alpha}}{x_1}-\RM^1(\alpha)=0,
    \quad|\alpha|=M.
\end{aligned}
\end{equation}
The relation between $\bF$ and $\bw$ is
\begin{equation}
\begin{aligned}
f_{\alpha}=\sum_{|\beta|\le|\alpha|}(-1)^{|\alpha-\beta|}
    F_{\beta}\frac{\He_{\alpha-\beta}\left(\dfrac{\bu}{\sqrt{\theta}}\right)}{(\alpha-\beta)!}\theta^{|\alpha-\beta|/2},\\
    u_i=\frac{F_{e_i}}{F_0},\quad
    p_{2e_i}=2F_{2e_i}-\frac{F_{e_i}^2}{F_0},\quad
    i\in\mathcal{D},
\end{aligned}
\end{equation}
where
$\He_{\alpha}\left(\dfrac{\bu}{\sqrt{\theta}}\right)=\prod_{d=1}^D
\He_{\alpha_d}\left(\dfrac{u_d}{\sqrt{\theta}}\right)$ and
$\He_{\alpha}(\bx)=0$ if at least one $\alpha_j$ is negative. In
addition, $\hat{F}_{\alpha}$ is as
\begin{equation}
\hat{F}_{\alpha}=\sum_{|\beta|\le|\alpha|}
(-1)^{|\alpha-\beta|}
    F_{\beta}\frac{\He_{\alpha+e_1-\beta}\left(\dfrac{\bu}{\sqrt{\theta}}\right)}{(\alpha+e_1-\beta)!}\theta^{|\alpha+e_1-\beta|/2}.
\end{equation}
For convenience, the quasi-linear form of \eqref{eq:conservation_F}
is written as
\begin{equation}\label{eq:halfconservationsystem}
\pd{\bF}{t}+\boldsymbol{\Gamma}(\bF)\pd{\bF}{x_1}=0,
\end{equation}
where $\boldsymbol{\Gamma}(\bF)$ is an $N\times N$ matrix
and depends on \eqref{eq:conservation_F}.

Since \eqref{eq:halfconservationsystem} is not a conservative system,
we have to adopt the DLM theory \cite{Maso} to study the shock wave.
For a shock wave the two constant states $\bF^L$ and $\bF^R$ are
connected through a single jump discontinuity in a genuinely
non-linear field $\alpha$ travelling at the speed $S_{\alpha}$, and
the following two conditions apply
\begin{itemize}
\item
Generalized Rankine-Hugoniot condition:
\begin{equation}\label{eq:RHcondition}
\int_0^1 \left[
  S_{\alpha} \boldsymbol{\rm I} - \boldsymbol{\Gamma} \left( \bPhi(\nu;
                  \bF^L, \bF^R) \right)
\right] \frac{\partial \bPhi}{\partial
            \nu}(\nu;\bF^L,\bF^R) \dd \nu = 0,
\end{equation}
where $\bf I$ is the $N\times N$ identity matrix, and $\bPhi(\nu;
\bF^L, \bF^R)$ is a locally Lipschitz mapping satisfying
\begin{equation}
\bPhi(0; \bF^L, \bF^R) = \bF^L, \quad \bPhi(1; \bF^L, \bF^R) = \bF^R.
\end{equation}
We refer the readers to \cite{Maso} for details. 
\item
Entropy condition:
\begin{equation}\label{eq:shockwave_entropy}
\lambda_{i,k}(\bF^L) > S_{\alpha} > \lambda_{i,k}(\bF^R),
\end{equation}
where $i=\alpha_1$ and $k=M+1-|\hat{\alpha}|$.
\end{itemize}
For conservation laws, \eqref{eq:RHcondition} is the same as the
classical Rankine-Hugoniot condition. Thus the first $\mN((M-1)e_D)$
rows of \eqref{eq:RHcondition} are independent of $\bPhi$. This allows
us to analyze the properties of the shock waves without regarding the
form of $\bPhi$.

The first equation and the $(D+1)$-th equation of
\eqref{eq:RHcondition} are precisely as
\begin{align}
\rho^L u_1^L - \rho^R u_1^R &= S_{\alpha}(\rho^L - \rho^R),
    \label{eq:shockwave_first}\\
\rho^L (u_1^L)^2 + p_{2e_1}^L - \rho^R (u_1^R)^2 - p_{2e_1}^R &=
    S_\alpha(\rho^L u_1^L - \rho^R u_1^R).\label{eq:shockwave_second}
\end{align}
\begin{itemize}
\item If $\rho^L\neq\rho^R$, \eqref{eq:shockwave_first} and
  \eqref{eq:shockwave_second} give
\begin{subequations}\label{eq:shockwave_S}
\begin{align}
S_{\alpha}&=\frac{\rho^L u_1^L - \rho^R u_1^R}{\rho^L -
    \rho^R}\label{eq:shockwave_S1}\\
        &=\frac{\rho^L (u_1^L)^2 + p_{2e_1}^L - \rho^R (u_1^R)^2 -
            p_{2e_1}^R} {\rho^L u_1^L - \rho^R
                u_1^R}\label{eq:shockwave_S2}.
\end{align}
\end{subequations}
Substituting \eqref{eq:shockwave_S1} into \eqref{eq:shockwave_entropy} and
multiplying both sides with $(\rho^L-\rho^R)^2$, we get
\begin{subequations}\label{eq:shockwave_entropy_RL}
\begin{align}
\rho^L(u_1^L-u_1^R)(\rho^L-\rho^R) &>
\rC{i}{k}(\rho^L-\rho^R)^2\sqrt{\theta^R},\label{eq:shockwave_entropy_R}\\
\rho^R(u_1^L-u_1^R)(\rho^L-\rho^R) &<
\rC{i}{k}(\rho^L-\rho^R)^2\sqrt{\theta^L}\label{eq:shockwave_entropy_L}.
\end{align}
\end{subequations}
If $\rC{i}{k}>0$, \eqref{eq:shockwave_entropy_R} gives 
\begin{equation}\label{eq:shockwave_rhou}
(u_1^L-u_1^R)(\rho^L-\rho^R) > 0.
\end{equation} 
Therefore, we can divide \eqref{eq:shockwave_entropy_RL} by 
$(u_1^L-u_1^R)(\rho^L-\rho^R)$ and obtain
\begin{equation}
\frac{\rho^L}{\sqrt{\theta^R}} > 
    \frac{\rC{i}{k} (\rho^L - \rho^R)}{u_1^L - u_1^R} >
    \frac{\rho^R}{\sqrt{\theta^L}}.
\end{equation}
Thus we have
\begin{equation}\label{eq:shockwave_rhop}
(\rho^L)^2\theta^L>(\rho^R)^2\theta^R.
\end{equation}
Furthermore, \eqref{eq:shockwave_S} gives us the relation
\begin{equation}\label{eq:shockwave_rhopu}
(\rho^L-\rho^R)(p_{2e_1}^L-p_{2e_1}^R)=\rho^L\rho^R(u_1^L-u_1^R)^2.
\end{equation}
If $\rho^L<\rho^R$, \eqref{eq:shockwave_rhou} indicates $u_1^L<u_1^R$.
\eqref{eq:shockwave_rhop} can be written as $\rho^Lp^L>\rho^Rp^R$,
so we have $p^L>p^R$.
If $\rho^L>\rho^R$, \eqref{eq:shockwave_rhou} indicates $u_1^L>u_1^R$.
Summarizing these results, we get
\begin{equation}
\begin{aligned}
&\text{if } \rho^L<\rho^R, \text{ then } u_1^L<u_1^R \text{ and }
p^L>p^R,\\
&\text{if } \rho^L>\rho^R, \text{ then } u_1^L>u_1^R. 
\end{aligned}
\end{equation}

Analogously, if $\rC{i}{k}<0$, we have
\begin{equation}
(u_1^L-u_1^R)(\rho^L-\rho^R) < 0,\quad
(\rho^L)^2\theta^L<(\rho^R)^2\theta^R,
\end{equation}
and \eqref{eq:shockwave_rhopu} still holds. Hence we get that
\begin{equation}
\begin{aligned}
&\text{if } \rho^L>\rho^R, \text{ then } u_1^L<u_1^R \text{ and }
p^L<p^R,\\
&\text{if } \rho^L<\rho^R, \text{ then } u_1^L>u_1^R. 
\end{aligned}
\end{equation}
\item
If $\rho^L=\rho^R$, \eqref{eq:shockwave_first} and
\eqref{eq:shockwave_second} make that $u_1^L=u_1^R$ and
$p_{2e_1}^L=p_{2e_1}^R$, respectively.
Therefore, \eqref{eq:shockwave_entropy} is turned into
\begin{equation}\label{eq:shockwave_c2_p}
\rC{i}{k}\sqrt{\theta^L}>\rC{i}{k}\sqrt{\theta^R}.
\end{equation}
The following result is then attained
\begin{subequations}
\begin{align}
&\text{if } \rC{i}{k}>0, \text{ then } u_1^L=u_1^R,\quad
    p^L>p^R,\\
&\text{if } \rC{i}{k}<0, \text{ then } u_1^L=u_1^R,\quad
    p^L<p^R.
\end{align}
\end{subequations}
\end{itemize}

Now we summarize all our discussion on the entropy conditions of the
three types of elementary waves in the following theorem: 
\begin{theorem}
  For the Riemann problem \eqref{eq:splitRiemannProblem}, for the
  wave of the $\alpha$-th family, $\rC{i}{k}$, the
  macroscopic velocities and pressures on both sides of the wave have
  the relation with the type of the wave as in
  Table \ref{wave_type}, where $\rC{i}{k}$ corresponds to the
  eigenvalue $\lambda_{i,k}=u_1+\rC{i}{k}\sqrt{\theta}$, $i=\alpha_1$
  and $k=M+1-|\hat{\alpha}|$.
\begin{table}[ht]
\centering
\begin{tabular}{|l|c||c|}
\hline
Wave type   &   Eigenvalue & Velocity and Pressure    \\ [2mm]\hline
\multirow{2}*{Rarefaction wave} & 
    $\frac{\mathstrut}{\mathstrut}$$\rC{i}{k}>0$ & $u_1^L \le u_1^R$,
    $p^L<p^R$  \\[2mm]\cline{2-3}
    &   $\frac{\mathstrut}{\mathstrut}$$\rC{i}{k}<0$ & $u_1^L \le u_1^R$, $p^L>p^R$  \\  [2mm] \cline{2-3}
\hline\hline
\multirow{3}*{Shock wave}   &
    $\frac{\mathstrut}{\mathstrut}$$\rC{i}{k}>0$ & $u_1^L\le u_1^R$,
    $p^L>p^R$   \\[2mm]\cline{2-3}
    &   $\frac{\mathstrut}{\mathstrut}$$\rC{i}{k}<0$ & $u_1^L\le u_1^R$, $p^L<p^R$   \\  [2mm] \cline{2-3}
    &   $\frac{\mathstrut}{\mathstrut}$$\rC{i}{k}\neq0$              & $u_1^L>u_1^R$    \\  \hline\hline
\multirow{2}*{Contact discontinuity}    &
    $\frac{\mathstrut}{\mathstrut}$$\rC{i}{k}=0$              &
    $u_1^L=u_1^R$  \\[2mm] \cline{2-3}
    &   $\frac{\mathstrut}{\mathstrut}$$\rC{i}{k}\neq0$       & $u_1^L=u_1^R$, $p^L=p^R$ \\\hline
\end{tabular}
\caption{The relation between the type classification of elementary
    wave and the eigenvalue, macroscopic velocity and
        pressure.}\label{wave_type}
\end{table}
\end{theorem}


\section*{Acknowledgements}
This research was supported in part by the National Basic Research
Program of China (2011CB309704) and Fok Ying Tong Education and NCET
in China.


\section*{Appendix}

\appendix

\section{Collection of Notations}
We list below some of the notations used for convenience.

\renewcommand\arraystretch{1.3}
\begin{longtable}{lX}
$Q(f,f)$                           & The Boltzmann collision operator
                                     taken as zero in this paper.\\
$\rC{i}{k}$                        & The $i$-th zero of
                                     $\He_{k}(x)$.\\ 
$\alpha$                           & $D$-dimensional multi-index.\\
$\tal$                             & $D$-dimensional multi-index,
                                     defined as
                                         $\tal=\alpha-\alpha_1e_1$.\\
$\hat{\alpha}$                     & $(D-1)$-dimensional multi-index,
                                     defined by
                                         \eqref{eq:def_hatalpha}.\\
$\cS_{D,M}$                        & All the $D$-dimensional
                                     multi-index satisfying the sum of
                                     components not more than $M$,
                                     defined by \eqref{eq:def_SDM}.\\
$\cS_{D,M}(\hat{\alpha})$          & The subset of $\cS_{D,M}$, defined
                                     by \eqref{eq:def_SDM}.\\
$\mN(\alpha)$                      & The ordinal number of $\alpha$ in
                                     $\cS_{D,M}$, which is permuted by
                                     lexicographic order, defined by
                                     \eqref{eq:def_mN}.\\
$N$                                & Cardinality of $\cS_{D,M}$.\\
$\bw$                              & The basic moments in the moment
                                     system, defined by \eqref{eq:basic_moments}.\\
$\bAM$                             & The coefficient matrix of the
                                     moment system, defined by
                                     \eqref{eq:1d}.\\
${\bf I}$                          & Identity matrix.\\
$\RM$                              & Regularization term based on the
                                     characteristic speed
                                     correction.\\
$\RM^{j}$                          & Regularization term in
                                     $x_j$ direction.\\
$\bhAM$                            & The regularized matrix of $\bAM$,
                                     defined by
                                     \eqref{eq:definition_bhAM}.\\
${\bf\Lambda}$                     & The diagonal matrix, used to
                                     make $\bhAM$ dimensionless,
                                     defined by
                                         \eqref{eq:def_Lambda}.\\
$\btAM$                            & The dimensionless matrix of
                                     $\bhAM-u_1{\bf I}$.\\
$\hat{p}_{e_i+e_k}, g_{\alpha}$    & The dimensionless of
                                     $p_{e_i+e_k}$, $f_{\alpha}$.\\
$\br$                              & The right eigenvector of
                                     $\btAM$, defined by
                                     \eqref{eq:def_br}.\\
$\bv$                              & The components of it is a subset
                                     of that of $\br$, and
                                     $v_{\mathcal{N}_{D-1}(\hat{\alpha})}=r_{w_{\mN(\tal)}}$.\\
$\bB$                              & The coefficient matrix defined in
                                     Lemma \ref{lem:numberofeigenvectors}.\\
$\hat{\bB}$                        & The block diagonal matrix, used to
                                     make $\bB$ a unit lower
                                     triangular matrix, defined in
                                     Lemma \ref{lem:numberofeigenvectors}.\\
$\bv^{(j)}$                        & The basis of $\bv$, defined by
                                     \eqref{eq:eigenvector_value2}.\\
$\br_{\hat{\alpha},i}$             & The eigenvector of $\btAM$,
                                     depended on $\bv^{(j)}$,
                                     $j=\mathcal{N}_{D-1}(\hat{\alpha})$,
                                     for the eigenvalue $\rC{i}{k}$,
                                       $k=M+1-|\hat{\alpha}|$.\\
$\tilde{\mathcal{P}}_{D,M}$       & The characteristic polynomial of
                                    $\btAM$, defined by
                                    \eqref{eq:eigenpolynomials}.\\
$\lambda_{i,k}$                   & The eigenvalue of $\bhAM$,
                                    $\lambda_{i,k}=u_1+\rC{i}{k}\sqrt{\theta}$.\\
$\hat{\br}_{\hat{\alpha},i}$      & The eigenvector of $\bhAM$,
                                    corresponding to
                                        $\br_{\hat{\alpha},i}$, 
                                    for the eigenvalue $\rC{i}{k}$,
                                       $k=M+1-|\hat{\alpha}|$, defined
                                    by \eqref{eq:eigenvector_hA}.\\
$\mathcal{P}_{D,M}$               & The characteristic polynomial of
                                    $\bhAM$, defined by
                                    \eqref{eq:def_eigenpolynomial}.\\
${\bf G} = (g_{ij})_{D \times D}$ & The rotational matrix. \\
$\tilde{\bw}$                     & The rotated moments, defined by
                                    \eqref{eq:tilde_theta_u} and
                                    \eqref{eq:tilde_f_alpha}. \\
$\vartheta, \varphi$              & The Grad-type indices in $\mathcal{D}^m$. \\
$\varsigma(\cdot)$                & Conversion from multi-indices to
                                    Grad-type indices. \\
$\sigma(\cdot)$                   & Conversion from Grad-type indices
                                    to multi-indices. \\
$\Pi_g(\vartheta, \varphi)$       & Product of entries in the
                                    rotational matrix, defined by
                                    \eqref{eq:Pi_g}. \\
$\mathcal{A}_n^m$                 & The set of all $m$-permutations in
                                    $\{1, \cdots, n\}$, defined by
                                    \eqref{eq:cal_A}. \\
$\bbSigma(\alpha)$                & The set of all Grad-type indices
                                    whose corresponding multi-indices
                                    are $\alpha$, defined by
                                    \eqref{eq:bbSigma}.
\end{longtable}



\bibliographystyle{plain}
\bibliography{../article}
\end{document}